\documentclass[conference, 10pt, a4paper]{IEEEtran}

\usepackage{url}

\usepackage{hyperref}
\hypersetup{
  colorlinks   = true, 
  urlcolor     = black!50!red, 
  linkcolor    = black!50!brown, 
  citecolor   = black!50!brown 
}

\usepackage{amsmath}
\usepackage{amssymb}
\usepackage{amsthm}
\usepackage{verbatim}
\usepackage{bm}
\usepackage{color,graphicx,xcolor}
\usepackage{mdwtab}
\usepackage{mathtools,tikz}
\mathtoolsset{showonlyrefs} 
\usepackage{hhline}
\usepackage{pdfpages}
\usepackage{enumitem}

\usepackage[skip=0.4em]{caption}
\DeclareCaptionFormat{myformat}{\fontsize{8}{9}\selectfont#1#2#3}
\usepackage{subcaption}

\IEEEoverridecommandlockouts
\usepackage{cite}
\usepackage{amsmath,amssymb,amsfonts,amsthm,xspace}
\usepackage{algorithmic}
\usepackage{graphicx}
\usepackage{textcomp}

\usetikzlibrary{arrows,shapes.geometric}
\allowdisplaybreaks
\newtheorem{thm}{Theorem}
\newtheorem{defn}{Definition}
\newtheorem{lemma}[thm]{Lemma}

\newcommand{\subparagraph}[1]{\par {\em\underline{#1:}}}

\usepackage{pgffor}
\foreach \x in {a,...,z}{%
\expandafter\xdef\csname vec\x \endcsname{\noexpand\ensuremath{\noexpand\bm{\x}}}
}

\foreach \x in {A,...,Z}{%
\expandafter\xdef\csname vec\x \endcsname{\noexpand\ensuremath{\noexpand\bm{\x}}}
}

\foreach \x in {A,...,Z}{%
\expandafter\xdef\csname c\x \endcsname{\noexpand\ensuremath{\noexpand\mathcal{\x}}}
}

\foreach \x in {A,...,Z}{%
\expandafter\xdef\csname bb\x \endcsname{\noexpand\ensuremath{\noexpand\mathbb{\x}}}
}

\newcommand{\Q}{\bbQ}
\renewcommand{\P}{\bbP}
\newcommand{\decfun}{\ensuremath{\psi}}
\newcommand{\dec}{\ensuremath{\widehat{H}}}
\newcommand{\tvd}[2]{\ensuremath{\|#1 - #2\|_\textup{TV}}}
\newcommand{\prot}{\ensuremath{\pi}}

\newcommand{\Prob}{\text{P}}

\newcommand{\type}[1]{T_{#1}}
\newcommand{\types}[2]{\cP_{#1}^{#2}}
\newcommand{\typeclass}[3]{T_{#1}^{#2}(#3)}

\newcommand{\expreg}{\cR}

\newcommand{\expc}{\alpha}
\newcommand{\expp}{\beta}

\newcommand{\binfun}{\Lambda}
\newcommand{\fand}{f_{\mathsf{AND}}}
\newcommand{\dist}[1]{\mathsf{p}\left(#1\right)}
\DeclarePairedDelimiterX{\infdivx}[2]{(}{)}{%
  #1\;\delimsize\|\;#2%
}
\DeclarePairedDelimiterX{\infdivxcond}[3]{(}{)}{%
  #1\;\delimsize\|\;#2\;|\;#3%
}
\newcommand{\kl}{D\infdivx}
\newcommand{\klc}{D\infdivxcond}

\newtheorem{eg}[thm]{Example}
\newtheorem{claim}[thm]{Claim}

\usetikzlibrary{shapes,arrows,positioning,decorations.pathreplacing,calc}

\newcommand{\MM}[1]{{\textcolor{blue}{#1}}}

\def\BibTeX{{\rm B\kern-.05em{\sc i\kern-.025em b}\kern-.08em
    T\kern-.1667em\lower.7ex\hbox{E}\kern-.125emX}}
\begin{document}

\title{Private Two-Terminal Hypothesis Testing}

\author{\IEEEauthorblockN{Varun Narayanan}
\IEEEauthorblockA{TIFR, India}
\and
\IEEEauthorblockN{Manoj Mishra}
\IEEEauthorblockA{NISER, Bhubaneswar, HBNI}
\and
\IEEEauthorblockN{Vinod M. Prabhakaran}
\IEEEauthorblockA{TIFR, India}
}

\maketitle

\begin{abstract}

We study private two-terminal hypothesis testing with simple hypotheses where the privacy goal is to ensure that participating in the testing protocol reveals little additional information about the other user's observation when a user is told what the correct hypothesis is. We show that, in general, meaningful correctness and privacy cannot be achieved if the users do not have access to correlated (but, not common) randomness. We characterize the optimal correctness and privacy error exponents when the users have access to non-trivial correlated randomness (those that permit secure multiparty computation).

\end{abstract}


\section{Introduction}\label{sec:intro}

Hypothesis testing is a basic statistical inference problem with a long history~\cite{Pearson1900, Gosset1908, Fisher1925,NeymanP33}. Multi-terminal hypothesis testing, or distributed detection, where data is distributed in space, has also been widely studied, e.g.,~\cite{AhlswedeC86,Han87,Tsitsiklis93,HanA98} where the primary interest is the communication required to carry out hypothesis testing. More recently, there has been a renewed interest on this question, see, e.g.,~\cite{YHKim12,Wagner12,YHKim13,ZhangDJW13,Gunduz17,Wigger18,HanOW18,GDKR19,AcharyaCT19} and references therein. When data is distributedly observed by different users, a natural question of interest is whether inference can be carried out while providing some privacy for the users and what, if any, trade-offs exist between the accuracy of inference and privacy. Several recent works have explored this question mostly when there are a large number of users each of whom observe a small part of the data~\cite{DuchiJW13,KairouzBR16,Sheffet18,AcharyaCFT2019,AndoniMN18,GilaniBST19,AcharyaCT19,AliakbarpourDKR19,Gunduz18}.

We study two-terminal binary hypothesis testing~\cite{AhlswedeC86,Han87,HanA98} with simple hypotheses. Instead of restricting the amount of communication between the user as these works did, our focus will be on guaranteeing privacy to the two users. Our definition of privacy is inspired by the definition of secure multiparty (function) computation (MPC)~\cite{CramerDN15}. In 2-user MPC, the goal is for each user to learn little {\em additional} information as possible about the input and output of the other user than what the user can infer from its own input and output. In other words, the protocol reveals just enough information about the other user's data to compute the function, but not much more. i.e., compared to when the user is simply told the evaluation of the function, each user gains little additional knowledge about the other user's data. For the 2-user distributed binary hypothesis testing problem we study here, our privacy goal is to ensure that participating in the protocol reveals little additional information about the other user's observation when a user is told what the correct hypothesis is (see Definition~\ref{def:pbht}).

First we show that meaningful correctness and privacy cannot be achieved, in general, if the users do not have access to correlated (but, not just common) randomness. This is analogous to the well-known fact that two-user MPC is also impossible without such stochastic resources~\cite{Kus89,Bea89,MPR09}. Indeed, we demonstrate this by reducing two-user MPC of the binary {\sf AND} function to a two-terminal private binary hypothesis testing problem. 

Our main result is a trade-off between the optimal error and privacy error exponents in the setting where the users have access to any correlated randomness which permits MPC. We do this by effectively reducing the problem to MPC of the decision function. The optimal trade-off is thus the best trade-off possible when the users are simply given the output of the decision function by a genie (i.e., a trusted third party).

We note that Andoni et al.~\cite{AndoniMN18}, among other things, also studied two-terminal private hypothesis testing. Their definition of privacy is also inspired by MPC, but is different from ours. They deem a protocol private as long as it is a MPC of any decision function with a good probability of error performance. This may not allow comparison of the privacy of different protocols, e.g., two decision functions may have similar error performances, but, they may have different worst-case privacy performance measured in terms of the information the decision function itself reveals to a user, conditioned on its observation, about the other user's observation. Here, we define privacy from first principles and arrive at MPC as a means to achieve it. Our definition also allows us to compare the privacy of protocols and obtain the optimal trade-off between the error and privacy error exponents.

\section{Notation}\label{sec:notation}

When $X$ is a random variable and $E$ is an event, $\dist{X | E}$ represents the distribution induced by $X$ conditioned on event $E$.
When $(X, Y)$ is jointly distributed, and $E$ is an event, $\dist{X | Y, E}$ is a distribution over distributions on $\cX$, such that the distribution $\dist{X | Y = y, E}$ occurs with probability $\Prob(Y = y | E)$.
Total variation distance between distributions $\P_{XY}$ and $\Q_{XY}$ is denoted by $\tvd{\P_{XY}}{\Q_{XY}}$.

We use the method of types and follow the notation of Cover and Thomas~\cite{CoverT06} .
The \emph{type of a sequence} $x^n \in \cX^n$ is denoted by $\type{x^n}$.
The \emph{set of types of} $\cX^n$ is $\types{\cX}{n} = \{\type{x^n}: x^n \in \cX^n\}$.
Finally, for $\P_X \in \types{\cX}{n}$, \emph{type class of} $\P_X$, denoted by $\typeclass{\cX}{n}{\P_X}$, is the set of all sequences of type $\P_X$, \emph{i.e.,} $\typeclass{\cX}{n}{\P_X} = \{x^n \in \cX^n : \type{x^n} = \P_X\}$.

\section{Problem Statement}

In a distributed binary hypothesis testing problem of sample size $n$, Alice and Bob observe $X^n$ and $Y^n$, respectively, where $(X^n,Y^n)$ is drawn independent and identically distributed (i.i.d.) according to the distribution $\P^0_{XY}$ under the null hypothesis $\Theta = 0$ and $\P^1_{XY}$ under the alternate hypothesis $\Theta = 1$.
Here, random variable $\Theta \in \{0,1\}$ represents the true hypothesis.
Independent of the observations, Alice and Bob have access to potentially dependent random variables $W_A,W_B$, respectively.
Alice and Bob engage in an interactive communication protocol $\prot_n$ in which they take turns exchanging messages with each other.
Messages produced by each user is a function of their observation, randomness, and messages exchanged so far.
At the end of the protocol, both users output their decision.
Let $V_A$ (resp., $V_B$) denote the {\em view} of Alice (resp., Bob) at the end of the protocol; this is the collection of the observations $X^n$, randomness $W_A$ and the transcript.
Let the decision function of Alice (resp., Bob) be denoted by $\decfun_A:\cV_A\rightarrow\{0,1\}$ (resp., $\decfun_B:\cV_B\rightarrow\{0,1\}$) and the decision itself by $\dec_A=\decfun_A(V_A)$ (resp., $\dec_B$).
Note that for $\prot_n$ to be a valid protocol, it must satisfy the natural conditional independence statement that each message produced by a user must be a function of what the users knows when it is produced.  
In this exposition, we consider the \emph{honest but curious} model of security in which Alice and Bob are obliged to follow the protocol honestly but can be curious, in that, they might try to infer the other user's observation from their respective views.

The protocol naturally induces the following joint distribution $\dist{\Theta, X^n, Y^n, W_A, W_B, V_A, V_B, \dec_A, \dec_B}$.
For $i, j \in \{0, 1\}$, when $I$ is the indicator function, this distribution can be described as follows.
\begin{multline}\label{eqn:prot-dist}
	\Prob(\theta, x^n, y^n, w_A, w_B, v_A, v_B, \dec_A = i, \dec_B = j)\\
	= \Prob_{\Theta}(\theta) \cdot \Prob_{X^n, Y^n | \Theta}(x^n, y^n | \theta)	\cdot \Prob_{W_A, W_B}(w_A, w_B)\\
	\cdot \Prob_{\prot_n}(v_A, v_B | x^n, y^n, w_A, w_B)\\
	\cdot I(\decfun_A(v_A) = i)	\cdot I(\decfun_B(v_B) = j).
\end{multline}

\begin{defn}[Private Binary Hypothesis Testing]\label{def:pbht}
A protocol $\prot_n$ is said to be an {\em $(n, \delta, \mu)$-private distributed binary hypothesis testing protocol} if for a sample size $n$, the distribution induced by $\prot_n$ given in~\eqref{eqn:prot-dist} satisfies $\delta$-correctness and $\mu$-privacy conditions given below.

\paragraph{Correctness} For $\delta \ge 0$, $\prot_n$ is said to be $\delta$-correct if
\begin{align}
\Prob\left( \dec_A = 1 - \theta | \Theta = \theta \right) \leq \delta \label{eqn:corr-A},\\
\Prob\left( \dec_B = 1 - \theta | \Theta = \theta \right) \leq \delta \label{eqn:corr-B}.
\end{align}

\paragraph{Privacy} For $\mu \ge 0$, $\prot_n$ is said to be $\mu$-private if for $\theta = 0, 1$,
\begin{multline}
\Prob\left( \tvd{\dist{Y^n | x^n, \theta}}{\dist{Y^n | V_A, x^n, w_A, \theta}} \ge \mu \right) \leq \mu,\\ \forall\, x^n, w_A\label{eqn:priv-A}
\end{multline}
\begin{multline}
\Prob\left( \tvd{\dist{X^n | y^n, \theta}}{\dist{X^n | V_B, y^n, w_B, \theta}} \ge \mu \right) \leq \mu,\\ \forall\, y^n, w_B. \label{eqn:priv-B}
\end{multline}
\end{defn}
\begin{defn}[{\em Weakly} Private Binary Hypothesis Testing]\label{def:weakprivacy}
A protocol $\prot_n$ is said to be an {\em $(n, \delta, \mu)$-weakly private distributed binary hypothesis testing protocol} if it satisfies $\delta$-correctness conditions ~\eqref{eqn:corr-A}~and~\eqref{eqn:corr-B}, and $\mu$-weak privacy conditions given below.
\paragraph{Weak privacy} For $\mu \ge 0$, $\prot_n$ is said to be $\mu$-weakly private if for $\theta = 0, 1$,
\begin{align}
\Prob\left( \tvd{\dist{Y^n | X^n, \theta}}{\dist{Y^n | V_A, X^n, \theta}} \ge \mu \right) \leq \mu,\label{eqn:weak-priv-A}\\
\Prob\left( \tvd{\dist{X^n | Y^n, \theta}}{\dist{X^n | V_B, Y^n, \theta}} \ge \mu \right) \leq \mu.\label{eqn:weak-priv-B}
\end{align}
\end{defn}
It is clear that an $(n, \delta, \mu)$-private protocol is also $(n, \delta, \mu)$-weakly private.
\begin{defn}
A pair $(\expc, \expp) \in \mathbb{R}^2_+$ is said to be \emph{achievable} if there exists a sequence of $(n,\delta_n,\mu_n)$-private protocols such that
\begin{align*}
	\limsup \limits_{n \rightarrow \infty} -\frac{1}{n} \log{\delta_n} &\ge \expc,\\
	\limsup \limits_{n \rightarrow \infty} -\frac{1}{n} \log{\mu_n} &\ge \expp.
\end{align*}
The closure of the set of all achievable pairs is the {\em correctness-privacy error exponent region} $\expreg$.
\end{defn}

\section{Results}\label{sec:results}

If Alice and Bob do not have access to correlated random variables, i.e., $W_A,W_B$ are independent, we argue using the following example that meaningful private hypothesis testing may not be possible. This will imply that the same holds true even if users share common randomness in addition since the honest-but-curious users may share part of their private randomness at the outset.
\begin{eg}\label{eg}
	Let $X, Y$ be distributed \emph{i.i.d.} $\sim \text{Bernoulli}(\frac{1}{2})$ under the null hypothesis $(\Theta = 0)$, and $X = Y \sim \text{Bernoulli}(\frac{1}{2})$ under the alternate hypothesis $(\Theta = 1)$.
\end{eg}
For this example, we show that it is impossible to realize even weakly private hypothesis testing for small correctness and privacy error using arbitrarily large number of samples.
We do this by providing a black box reduction of the statistically secure 2-party computation of a function, which is known to be impossible, to weakly private hypothesis testing protocol.
This proof can be extended to show the impossibility of weakly private hypothesis testing of independence in general, \emph{i.e.,} the null hypothesis is a joint distribution $\P_{XY}$ and the alternate hypothesis the independent distribution $\P_X \cdot \P_Y$.
The result is formally stated in the following theorem.

\begin{thm}\label{thm:no-setup-impossibility}
	For the hypothesis testing problem described in Example~\ref{eg}, $(n, \delta, \mu)$-weakly private distributed hypothesis testing is impossible for all $n \in \mathbb{N}$, when $\delta, \mu \le \frac{1}{12}$.
\end{thm}

To enable private hypothesis testing we will assume in the sequel that Alice and Bob have access to independent copies of non-trivially correlated random variables $W_A, W_B$.
By non-trivial correlations we mean the large class of correlations~\cite{Kus89,Bea89,MPR09} that are complete for secure 2-party computation. Without loss of generality, we consider $W_A, W_B$ which are independent copies of oblivious transfer correlations (Definition~\ref{def:ot}) in our positive results.

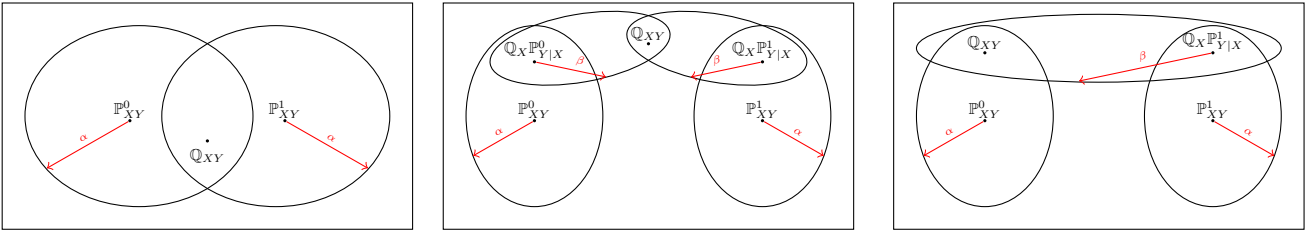
\begin{figure*}[h]

\begin{subfigure}{.32\textwidth}
\setlength{\unitlength}{1cm}
\centering
\begin{tikzpicture}[scale=0.6, every node/.style={scale=0.6}]

\def\x{0}
\def\y{0}
\def\ht{5}
\def\len{9}
\def\Ctx{\x+3}
\def\CtxSh{\x+6}
\def\Cty{\y+2}
\def\rx{2.5}
\def\ry{2}
\def\linelen{2.1}

\def\P{\mathbb{P}^0_{XY}}
\def\Q{\mathbb{P}^1_{XY}}
\def\T{\mathbb{Q}_{XY}}
\def\linelabel{\scriptstyle{\alpha}}

\draw (\x,\y-0.5) rectangle ++(\len,\ht);
\draw (\Ctx,\Cty) circle [x radius=\rx, y radius=\ry];
\draw (\CtxSh,\Cty) circle [x radius=\rx, y radius=\ry];

\node (H0) at (\Ctx - 0.2,\Cty) {$\P$}; 
\node (H1) at (\CtxSh + 0.2,\Cty) {$\Q$}; 
\node (D0) at (\Ctx+1.5,\Cty-1) {$\T$}; 

\node[circle,fill=black,inner sep=0pt,minimum size=2pt,below=1pt] (H0Dot) at (H0) {}; 
\node[circle,fill=black,inner sep=0pt,minimum size=2pt,below=1pt] (H1Dot) at (H1) {}; 
\node[circle,fill=black,inner sep=0pt,minimum size=2pt,above=7pt] (D0Dot) at (D0) {}; 

\draw[->,red] (H0Dot) -- node[above, sloped] {$\linelabel$} +(210:\linelen);
\draw[->,red] (H1Dot) -- node[above, sloped] {$\linelabel$} +(330:\linelen);

\end{tikzpicture}

\end{subfigure}%
\begin{subfigure}{.32\textwidth}


\setlength{\unitlength}{1cm}
\centering
\begin{tikzpicture}[scale=0.6, every node/.style={scale=0.6}]

\def\x{0}
\def\y{0}
\def\ht{5}
\def\len{9}
\def\Ctx{\x+2}
\def\CtxSh{\x+7}
\def\Cty{\y+2}
\def\CtxTlt{\x+3}
\def\CtxTltSh{\x+6}
\def\CtyTlt{\y+3.5}
\def\rx{1.5}
\def\ry{2}
\def\rxTlt{2}
\def\ryTlt{0.75}

\def\P{\mathbb{P}^0_{XY}}
\def\Q{\mathbb{P}^1_{XY}}
\def\TP{\mathbb{Q}_X \mathbb{P}^0_{Y|X}}
\def\TQ{\mathbb{Q}_X \mathbb{P}^1_{Y|X}}
\def\T{\mathbb{Q}_{XY}}

\def\Pxloc{2}
\def\Pyloc{1.3}

\def\linelen{1.55}
\def\linelabel{\scriptstyle{\alpha}}

\def\linelenNew{1.6}
\def\linelabelNew{\scriptstyle{\beta}}

\draw (\x,\y-0.5) rectangle ++(\len,\ht);
\draw (\Ctx,\Cty) circle [x radius=\rx, y radius=\ry];
\draw (\CtxSh,\Cty) circle [x radius=\rx, y radius=\ry];
\draw (\CtxTlt,\CtyTlt) circle [x radius=\rxTlt, y radius=\ryTlt, rotate=10];
\draw (\CtxTltSh,\CtyTlt) circle [x radius=\rxTlt, y radius=\ryTlt, rotate=-10];

\node (H0) at (\Ctx,\Cty) {$\P$}; 
\node (H1) at (\CtxSh,\Cty) {$\Q$}; 
\node (M1) at (\Pxloc+5,\Pyloc+2) {$\TQ$}; 
\node (M0) at (\Pxloc,\Pyloc+2) {$\TP$}; 
\node (D0) at (\Pxloc+2.5,\Pyloc+2.4) {$\T$}; 

\node[circle,fill=black,inner sep=0pt,minimum size=2pt,below=1pt] (H0Dot) at (H0) {}; 
\node[circle,fill=black,inner sep=0pt,minimum size=2pt,below=1pt] (H1Dot) at (H1) {}; 
\node[circle,fill=black,inner sep=0pt,minimum size=2pt,below=1pt] (M0Dot) at (M0) {}; 
\node[circle,fill=black,inner sep=0pt,minimum size=2pt,below=1pt] (M1Dot) at (M1) {}; 
\node[circle,fill=black,inner sep=0pt,minimum size=2pt,below=1pt] (D0Dot) at (D0) {}; 

\draw[->,red] (H0Dot) -- node[above, sloped] {$\linelabel$} +(210:\linelen);
\draw[->,red] (H1Dot) -- node[above, sloped] {$\linelabel$} +(330:\linelen);
\draw[->,red] (M0Dot) -- node[above, sloped, pos=0.6] {$\linelabelNew$} +(348:\linelenNew);
\draw[->,red] (M1Dot) -- node[above, sloped, pos=0.6] {$\linelabelNew$} +(192:\linelenNew);

\end{tikzpicture}

\end{subfigure}
\begin{subfigure}{.32\textwidth}


\setlength{\unitlength}{1cm}
\centering
\begin{tikzpicture}[scale=0.6, every node/.style={scale=0.6}]

\def\x{0}
\def\y{0}
\def\ht{5}
\def\len{9}
\def\Ctx{\x+2}
\def\CtxSh{\x+7}
\def\Cty{\y+2}
\def\CtxTlt{\x+4.5}
\def\CtyTlt{\y+3.5}
\def\rx{1.5}
\def\ry{2}
\def\rxTlt{4}
\def\ryTlt{0.75}

\def\P{\mathbb{P}^0_{XY}}
\def\Q{\mathbb{P}^1_{XY}}
\def\TQ{\mathbb{Q}_X \mathbb{P}^1_{Y|X}}
\def\T{\mathbb{Q}_{XY}}

\def\Pxloc{2}
\def\Pyloc{1.5}

\def\linelen{1.55}
\def\linelabel{\scriptstyle{\alpha}}

\def\linelenNew{3}
\def\linelabelNew{\scriptstyle{\beta}}

\draw (\x,\y-0.5) rectangle ++(\len,\ht);
\draw (\Ctx,\Cty) circle [x radius=\rx, y radius=\ry];
\draw (\CtxSh,\Cty) circle [x radius=\rx, y radius=\ry];
\draw (\CtxTlt,\CtyTlt) circle [x radius=\rxTlt, y radius=\ryTlt, rotate=0];

\node (H0) at (\Ctx,\Cty) {$\P$}; 
\node (H1) at (\CtxSh,\Cty) {$\Q$}; 
\node (M1) at (\Pxloc+5,\Pyloc+2) {$\TQ$}; 
\node (M0) at (\Pxloc,\Pyloc+2) {$\T$}; 

\node[circle,fill=black,inner sep=0pt,minimum size=2pt,below=1pt] (H0Dot) at (H0) {}; 
\node[circle,fill=black,inner sep=0pt,minimum size=2pt,below=1pt] (H1Dot) at (H1) {}; 
\node[circle,fill=black,inner sep=0pt,minimum size=2pt,below=1pt] (M0Dot) at (M0) {}; 
\node[circle,fill=black,inner sep=0pt,minimum size=2pt,below=1pt] (M1Dot) at (M1) {}; 

\draw[->,red] (H0Dot) -- node[above, sloped] {$\linelabel$} +(210:\linelen);
\draw[->,red] (H1Dot) -- node[above, sloped] {$\linelabel$} +(330:\linelen);
\draw[->,red] (M1Dot) -- node[above, sloped] {$\linelabelNew$} +(192:\linelenNew);

\end{tikzpicture}

\end{subfigure}%
\caption{From left to right, the diagrams illustrate the conditions (i), (ii) and (iv) (for $\theta = 0$) in Theorem~\ref{thm:error-exponent}, respectively.
In the figure, for $\theta = 0, 1$, the set around $\P^{\theta}_{XY}$ encloses all distributions $\mathbb{T}_{XY}$ s.t. $\kl{\mathbb{T}_{XY}}{\P^{\theta}_{XY}} \le \expc$, and the set around $\Q_X \cdot \P^{\theta}_{Y|X}$ encloses all conditional distributions $\mathbb{T}_{Y|X}$ s.t. $\klc{\mathbb{T}_{Y|X}}{\P^{\theta}_{Y|X}}{\Q_X} \le \beta$.
Note that, for any $\Q_X$, $\kl{\Q_X \cdot \P^{\theta}_{Y|X}}{\P^{\theta}_{XY}} = \kl{\Q_X}{\P^{\theta}_X}$.}
\label{fig}
\end{figure*}


Our main result is the following characterization of the correctness-privacy error exponent region:
\begin{thm}\label{thm:error-exponent}
For the binary distributed hypothesis testing problem, correctness-privacy error exponent $(\expc, \expp)$ is achievable if and only if the following conditions are satisfied.
There exists no distribution $\Q_{XY}$ for which any of the following conditions hold,
\begin{enumerate}
\item[(i).] 
\begin{align}
	\kl{\Q_{XY}}{\P^{\theta}_{XY}} \le \expc, \forall \theta \in \{0, 1\}.\label{eqn:char-0}
\end{align}
\item[(ii).] 
\begin{align}
	\kl{\Q_X}{\P^\theta_X} & \le \expc, \nonumber\\
	\text{and } \klc{\Q_{Y|X}}{\P^\theta_{Y|X}}{\Q_X} & \le \expp, \forall \theta \in \{0, 1\}. \label{eqn:char-2-A}
\end{align}
\item[(iii).] 
\begin{align}
	\kl{\Q_Y}{\P^\theta_Y} & \le \expc, \nonumber\\
	\text{and } \klc{\Q_{X|Y}}{\P^\theta_{X|Y}}{\Q_Y} & \le \expp,  \forall \theta \in \{0, 1\}. \label{eqn:char-2-B}
\end{align}
\item[(iv).] 
\begin{align}
	\kl{\Q_{XY}}{\P^\theta_{XY}} \le \expc, \kl{\Q_X}{\P^{1 - \theta}_X} \le \expc,\nonumber\\
	\text{ and } \klc{\Q_{Y|X}}{\P^{1 - \theta}_{Y|X}}{\Q_X} \le \expp, \forall \theta \in \{0, 1\}. \label{eqn:char-1-A}
\end{align}
\item[(v).] 
\begin{align}
	\kl{\Q_{XY}}{\P^\theta_{XY}} \le \expc, \kl{\Q_Y}{\P^{1 - \theta}_Y} \le \expc,\nonumber\\
	\text{ and } \klc{\Q_{X|Y}}{\P^{1 - \theta}_{X|Y}}{\Q_Y} \le \expp, \forall \theta \in \{0, 1\}. \label{eqn:char-1-B}
\end{align}
\end{enumerate}
\end{thm}
We prove this by effectively reducing the problem to MPC of a decision function.
The trade-off above is the best possible for any decision function.
The intuition behind the conditions in the theorem are as follows: See Figure~\ref{fig} which illustrates the conditions (i), (ii)~and~(iv).
Condition (i) simply follows from the error exponent for non-private hypothesis testing where the optimal error exponent can be obtained by deciding in favor of the hypothesis $\theta$ which minimizes $\kl{\Q_{XY}}{\P^{\theta}_{XY}}$ where $\Q_{XY}$ is the type of the observation.
In condition (ii), the observed type $Q_{XY}$ is such that $\kl{\Q_{X}}{\P^{\theta}_{X}}\le \expc$ for both $\theta=0,1$, i.e., to obtain the prescribed error exponent, Alice may not make a decision simply based on her observed vector.
Now, to ensure a privacy error exponent against Alice of at least $\expp$, for all observed conditional types $\Q_{Y|X}$ such that $\klc{\Q_{Y|X}}{\P^\theta_{Y|X}}{\Q_X} \le \expp$, her decision must be $\theta$.
This leads to (ii).
Condition (iv) arises from an interplay of Alice's correctness condition for one of the hypotheses and her privacy condition for the other.
As before, the observed type $Q_{XY}$ is such that $\kl{\Q_{X}}{\P^{\theta}_{X}}\le \expc$ for both $\theta=0,1$ and Alice may not make a decision only based on her observation.
Correctness condition for the hypothesis $\theta$ requires that, if $\kl{\Q_{XY}}{\P^\theta_{XY}} \le \expc$, she must decide in favor of $\theta$, but, privacy requires that, if $\klc{\Q_{Y|X}}{\P^{1 - \theta}_{Y|X}}{\Q_X} \le \expp$, her decision must be $1-\theta$.
This gives rise to (iv). 
Conditions (iii) and (v) are analogous to (ii) and (iv), respectively, from Bob's side.
\begin{figure}[htb]
\begin{center}
\includegraphics[height=1.5in,width=3in]{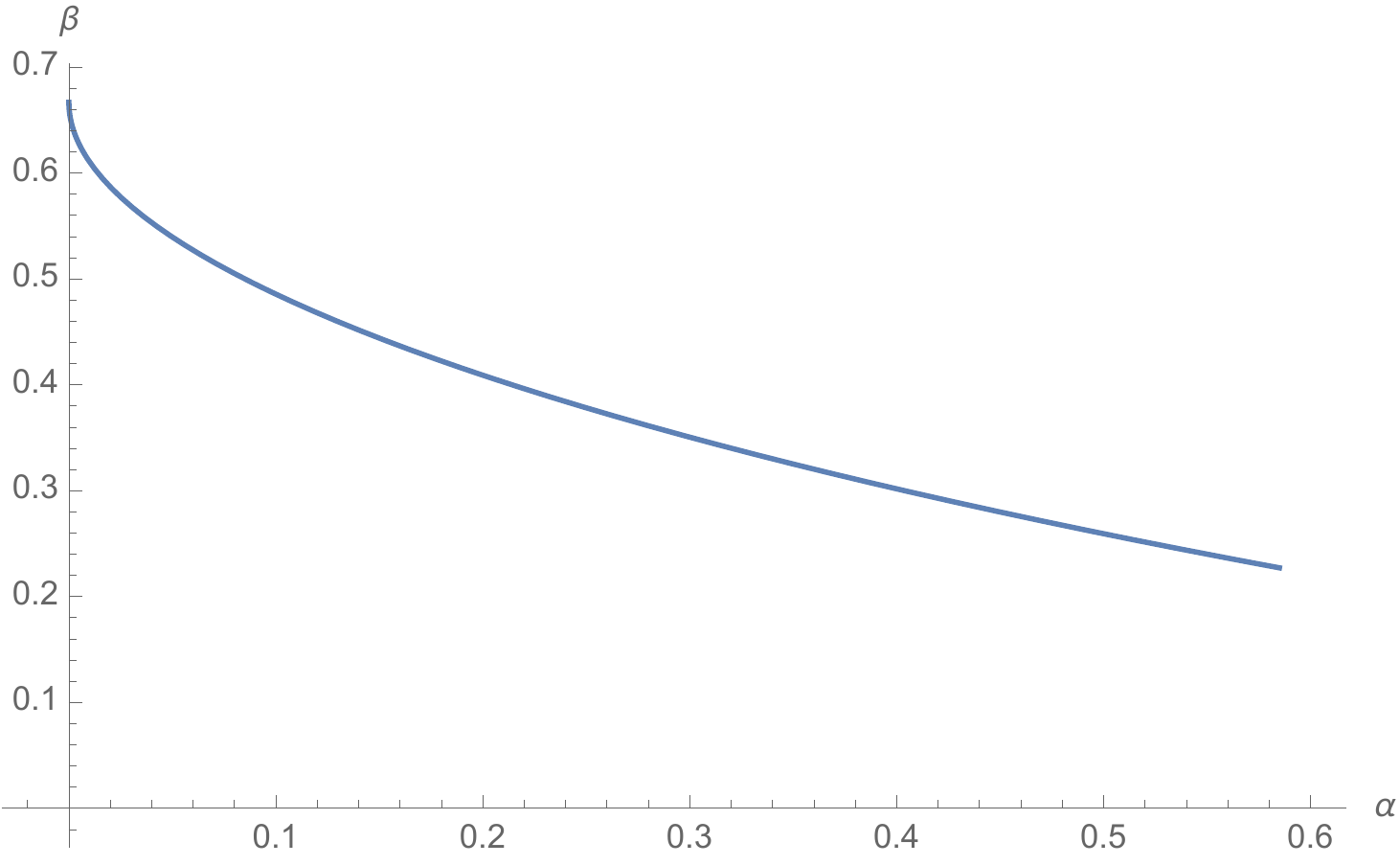}
\caption{The optimal correctness-privacy error exponent trade-off for the decision problem with null hypothesis: $\P^0(0,0) = \P^0(0,1) = \P^0(1,1) = \frac{1}{3}$ and alternate hypothesis: $\P^1(0,0) = \frac{2}{3}, \P^1(1, 1) = \frac{1}{3}$. Log is computed with base 2}
\end{center}
\end{figure}

\section{Proof Sketches}\label{sec:proofs}

In this section we provide the proof sketches of the results provided in Section~\ref{sec:results}.
The detailed proofs are provided in the Appendix.
Before we prove Theorem~\ref{thm:no-setup-impossibility}, we state the following well known result that shows the impossibility of secure computation of {\sf AND} function.
Let $\fand$ denote the {\sf AND} function, \emph{i.e.,} $\fand(x, y) = x \wedge y$ for $x, y \in \{0, 1\}$.

\begin{thm}\label{thm:and-impossible}\cite{Kilian00}
When Alice and Bob receive $u, v \in \{0, 1\}$, respectively, and have access to $W_A, W_B$, respectively, where $W_A, W_B$ are independent, it is impossible to compute $\fand(u, v)$ with $\frac{1}{6}$-security, \emph{i.e.,} there exists no protocol $\prot$ that satisfies the following properties.
\begin{align}
& \Prob_{\dec_A | U, V} (\dec_A \neq \fand(U, V) | u, v) \le \frac{1}{6},\label{eqn:mpc-corrA}\\
& \Prob_{\dec_B | U, V} (\dec_B \neq \fand(U, V) | u, v) \le \frac{1}{6},\label{eqn:mpc-corrB}\\
& \tvd{\dist{V_A | U = 0, V = 0}}{\dist{V_A | U = 0, V = 1}}\nonumber\\ & \qquad \qquad \qquad \qquad \qquad \qquad \qquad \qquad \qquad \quad \le \frac{1}{6},\label{eqn:mpc-privA}\\
& \tvd{\dist{V_B | U = 0, V = 0}}{\dist{V_B | U = 1, V = 0}}\nonumber\\ & \qquad \qquad \qquad \qquad \qquad \qquad \qquad \qquad \qquad \quad  \le \frac{1}{6}.\label{eqn:mpc-privB}\\
\end{align}
\end{thm}

\subsection{Proof Sketch for Theorem~\ref{thm:no-setup-impossibility}}\label{sec:eg}
	Given a $(n, \delta, \mu)$-weakly private protocol $\prot$ for the hypothesis testing problem described in Example~\ref{eg}, we construct a $\tau$-secure protocol $\prot_{\wedge}$ for computing $\fand$, where $\tau = \max(\delta, 2\mu)$.
	The impossibility of $\prot$ now follows from Theorem~\ref{thm:and-impossible} which states that such a $\prot_{\wedge}$ is impossible since $\tau = \frac{1}{6}$.

	\paragraph*{Description of $\prot_{\wedge}$} When Alice and Bob receive $u, v \in \{0, 1\}$, respectively as input.
	\begin{enumerate}
		\item Alice samples $Z^n$ uniformly from $\{0, 1\}^n$ and sends it to Bob.
		\item If $u = 1$, Alice sets $X^n = Z^n$, else $X^n = \hat{X}^n$, where $\hat{X}^n$ is uniform in $\{0, 1\}^n$ and independent of $Z^n$. 
		\item If $v = 1$, Bob sets $Y^n = Z^n$, else $Y^n = \hat{Y}^n$, where $\hat{Y}^n$ is uniform in $\{0, 1\}^n$ and independent of $Z^n$. 
		\item Alice and Bob execute $\prot$ with $X^n, Y^n$ as inputs, respectively, and output whatever $\prot$ outputs.
	\end{enumerate}
	\begin{claim}\label{clm:eg-corr}
		$\prot_{\wedge}$ computes $\fand$ with $\tau$-security.
	\end{claim}
	\begin{proof}[Proof sketch]
	When $\wedge(u, v) = 1$, inputs of Alice and Bob in $\prot$, \emph{viz.} $X^n, Y^n$, come according to the hypothesis $X^n = Y^n$ (\emph{i.e.,} alt. hypothesis, $\Theta = 1$), and when $\wedge(u, v) = 0$, $X^n$ and $Y^n$ are independent (\emph{i.e.,} null hypothesis, $\Theta = 0$).
	That $\prot_\wedge$ satisfies condition~\eqref{eqn:mpc-corrA}~and~\eqref{eqn:mpc-corrB}, now follows from the $\delta$-correctness of $\prot$.

	Next we show that $\prot_\wedge$ satisfies the condition~\eqref{eqn:mpc-privA}; that it satisfies~\eqref{eqn:mpc-privB} can be shown similarly.
	
	When $(u, v) = (0, 0)$, Alice's view $V_A = (Z^n, V_A^{\prot}(\hat{X}^n, \hat{Y}^n))$, where $V_A^{\prot}(\hat{X}^n, \hat{Y}^n)$ is the view of Alice when $\prot$ is executed with $\hat{X}^n, \hat{Y}^n$ as inputs of Alice and Bob, respectively.
	Note that $V_A^{\prot}(\hat{X}^n, \hat{Y}^n)$ consists of $\hat{X}^n, W_A$ and the transcript of the protocol $\prot$.
	Similarly, when $(u, v) = (0, 1)$, $V_A = (Z^n, V_A^{\prot}(\hat{X}^n, Z^n))$.
	The following claim, proved in the Appendix, shows that the statistical distance between Alice's views in these two cases is at most $2 \mu$.
	\begin{claim}\label{clm:eg-priv}
		If $\prot$ is $\mu$-weakly private, when $Z^n, \hat{X}^n, \hat{Y}^n$ are independently and uniformly distributed in $\{0, 1\}^n$, 
		\begin{align*}
			\tvd{\dist{Z^n, V_A^{\prot} (\hat{X}^n, Z^n)}}{\dist{Z^n, V_A^{\prot}(\hat{X}^n, \hat{Y}^n)}} \le 2\mu.
		\end{align*}
	\end{claim}
	\end{proof}
	Claim~\ref{clm:eg-corr},~\ref{clm:eg-priv} and Theorem~\ref{thm:and-impossible} together imply the impossibility of $\prot$, proving the theorem.

\subsection{Proof Sketch for Theorem~\ref{thm:error-exponent} (Converse)}

\paragraph*{Overview} The proof proceeds in two steps.
In Lemma~\ref{lem:binfun-type}, we show that if a protocol has small correctness and privacy error, when Alice and Bob observe only the input and output of the decision function (as if the users were simply given the outputs of the decision functions by a genie), correctness and privacy error is small with respect to an average notion of privacy.
In the second step, we show the converse by providing an upper bound on correctness-privacy error exponent region of such functions (Lemma~\ref{lem:binfun-seq}) with respect to the above mentioned notion of privacy.

In the sequel, for brevity, we will often represent $\binfun_A(X^n, Y^n)$ (resp. $\binfun_B(X^n, Y^n)$) by $\binfun_A$ (resp. $\binfun_B$) whenever it does not cause confusion.

\begin{lemma}\label{lem:binfun-type}
Given a $(n, \delta, \mu)$-private binary hypothesis testing protocol $\prot$, let $\binfun_A, \binfun_B$ be randomized boolean functions such that for all $x, y$,
\begin{align}\label{eqn:protfun-def}
	\dist{\dec_A | x, y} \equiv \dist{\binfun_A(X^n, Y^n) | x, y},\\
	\dist{\dec_B | x, y} \equiv \dist{\binfun_B(X^n, Y^n) | x, y},
\end{align}
where $(\equiv)$ denotes identical distributions. Then,
\begin{enumerate}
\item[(i).] $\Prob_{\binfun_A | \Theta}(1 - \theta | \theta) \le \delta$ and $\Prob_{\binfun_B | \Theta}(1 - \theta | \theta) \le \delta$
\item[(ii).] For $\theta = 0, 1$, for all $x^n \in \cX^n$,
\begin{align}
&\sum_{i \in \{0, 1\}} \Prob_{\binfun_A | X^n, \Theta} (i | x^n, \theta) \nonumber \\
& \qquad \cdot \tvd{\dist{Y^n|x^n, \theta}}{\dist{Y^n|\binfun_A = i, x^n, \theta}} \le 2\mu.\label{eqn:binfun-type-priv-A}
\end{align}
For $\theta = 0, 1$, for all $y^n \in \cY^n$,
\begin{align}
&\sum_{i \in \{0, 1\}} \Prob_{\binfun_B | Y^n, \Theta} (i | y^n, \theta)\nonumber \\
& \qquad \cdot \tvd{\dist{X^n|y^n, \theta}}{\dist{X^n|\binfun_B = i, y^n, \theta}} \le 2\mu.\label{eqn:binfun-type-priv-B}
\end{align}
\end{enumerate}
\end{lemma}
\begin{proof}[Proof sketch]
	Essentially, $\binfun_A, \binfun_B$ are the decision functions computed by Alice and Bob, respectively, in the protocol $\prot$.
	The statement about correctness (i) directly follows from this observation.
	For $\theta = 0, 1$ and for all $x^n \in \cX^n$,
	The privacy statement (ii) for $\binfun_A$ can be shown using a simple averaging argument on $\mu$-privacy of $\prot$.
	The privacy statement for $\binfun_B$ can be shown similarly.
\end{proof}

We now proceed to the second step of the proof.
If correctness-privacy error exponent $(\expc, \expp)$ is achievable, then there exists a sequence of protocols $\left(\prot_{n_i}\right)_{i \in \mathbb{N}}$ such that, $\prot_{n_i}$ is a $(n_i, \delta_i, \mu_i)$-private binary hypothesis testing protocol such that,
\begin{align*}
	\lim \limits_{i \rightarrow \infty} -\frac{1}{n_i} \log{\delta_{n_i}} \ge \expc,
	\lim \limits_{i \rightarrow \infty} -\frac{1}{n_i} \log{\mu_{n_i}} \ge \expp.
\end{align*}

Appealing to Lemma~\ref{lem:binfun-type}, for each $n_i$, we construct boolean randomized functions $(\binfun_A^{n_i}, \binfun_B^{n_i})$ from $\prot_{n_i}$.
Since $\prot_{n_i}$ is a $(n_i, \delta_i, \mu_i)$-private binary hypothesis testing protocol, $(\binfun_A^{n_i}, \binfun_B^{n_i})$ satisfy conditions (i), (ii) in Lemma~\ref{lem:binfun-type} w.r.t. the parameters $\delta_i$ and $\mu_i$.
In the next lemma, we show the necessity of conditions (i), (ii)~and~(iv) in Theorem~\ref{thm:error-exponent} using the sequence of functions, $\left(\binfun_A^{n_i}\right)_{i \in \mathbb{N}}$. 
The necessity of conditions (iii)~and~(v) can be shown similarly by analyzing $\left(\binfun_B^{n_i}\right)_{i \in \mathbb{N}}$. 
Thus, it remains to prove the following lemma.
\begin{lemma}\label{lem:binfun-seq}
Let $(\binfun_A^{n_i})_{i \in \mathbb{N}}$ be a sequence of randomized boolean functions that satisfy the following properties for each $n_i$.
\begin{enumerate}
\item[(i).] $\Prob_{\binfun^{n_i}_A(X^{n_i}, Y^{n_i}) | \Theta}(1 - \theta | \theta) \le \delta_{n_i}$ for $\theta = 0, 1$,
\item[(ii).]For $\theta = 0, 1$, for all $x^{n_i} \in \cX^{n_i}$,
\begin{multline*}
\sum_{i = 0, 1} \Prob_{\binfun_A^{n_i} | X^{n_i}, \Theta}(i | x^{n_i}, \theta) \cdot\\
\tvd{\dist{Y^{n_i}|x^{n_i}, \theta}}{\dist{Y^{n_i}|\binfun_A^{n_i}, x^{n_i}, \theta}}\\ \le 2\mu_{n_i}.
\end{multline*}
\end{enumerate}
If $\expc, \expp$ are such that,
\begin{align*}
	\expc \le \lim \limits_{i \rightarrow \infty} -\frac{1}{n_i} \log{\delta_{n_i}}, \;
	\expp \le \lim \limits_{i \rightarrow \infty} -\frac{1}{n_i} \log{\mu_{n_i}},
\end{align*}
then $(\expc, \expp)$ must satisfy conditions (i), (ii)~and~(iv) in Theorem~\ref{thm:error-exponent}.
\end{lemma}
\begin{proof}[Proof sketch]
	The following two claims are proved in the Appendix.
	\begin{claim}\label{clm:binfun-seq-corr}
		For $\theta = 0, 1$, for large enough $n$ in the sequence $(n_i)_{i \in \mathbb{N}}$, if $\Q_{XY} \in \types{\cX \times \cY}{n}$ and $\kl{\Q_{XY}}{\P^\theta_{XY}} < \expc$, then $\Prob(X^n \in S | X^n \in \Q_X, \Theta = \theta) \ge \frac{99}{100}$, where, 
		\begin{multline}
			S = \{x^n \in \typeclass{\cX}{n}{\Q_X} : \\ \frac{\sum_{y^n : (x^n, y^n) \in \typeclass{\cX \times \cY}{n}{\Q_{XY}}} \Prob(\binfun_A = \theta | x^n, y^n)}{|\{y^n : (x^n, y^n) \in \typeclass{\cX \times \cY}{n}{\Q_{XY}}\}|} \ge \frac{99}{100}\}.\label{eqn:binfun-seq-corr-set}
		\end{multline}
	\end{claim}
	\begin{claim}\label{clm:binfun-seq-priv}
		For $\theta = 0, 1$, for large enough $n$ in the sequence $(n_i)_{i \in \mathbb{N}}$, if $\Q_{XY} \in \types{\cX \times \cY}{n}, \kl{\Q_X}{\P^\theta_X} < \expc$, and $\klc{\Q_{Y|X}}{\P^\theta_{Y|X}}{\Q(x)} < \expp$, then $\Prob(X^n \in S | X^n \in \Q_X, \Theta = \theta) \ge \frac{99}{100}$, where, 
		\begin{multline}
			S = \{x^n \in \typeclass{\cX}{n}{\Q_X} : \\ \frac{\sum_{y^n : (x^n, y^n) \in \typeclass{\cX \times \cY}{n}{\Q_{XY}}} \Prob(\binfun_A = \theta | x^n, y^n)}{|\{y^n : (x^n, y^n) \in \typeclass{\cX \times \cY}{n}{\Q_{XY}}\}|} \ge \frac{80}{100}\}.\label{eqn:binfun-seq-priv-set}
		\end{multline}
	\end{claim}
	If there exists $\Q_{XY}$ that satisfies the inequalities in~\eqref{eqn:char-0}, then by Claim~\ref{clm:binfun-seq-corr}, for large enough $n$, Condition~\eqref{eqn:binfun-seq-corr-set} would be satisfied for $\theta = 0$ and $1$ for some $x^n \in \Q_X$; a contradiction.
	This proves the necessity of condition (i) in the theorem.

	If there exists $\Q_{XY}$ that satisfies the inequalities in~\eqref{eqn:char-2-A}, by Claim~\ref{clm:binfun-seq-priv}, for large enough $n$, Condition~\eqref{eqn:binfun-seq-priv-set} would be satisfied for $\theta = 0$ and $\theta = 1$ for some $x^n \in \Q_X$; a contradiction.
	This proves the necessity of Condition (ii).

	If there exists $\Q_{XY}$ that satisfies the inequalities in~\eqref{eqn:char-1-A}, then for large enough $n$, by Claim~\ref{clm:binfun-seq-corr}~and~\eqref{eqn:binfun-seq-corr-set} would be satisfied for $\theta$ and $1 - \theta$, respectively, by Claim~\ref{clm:binfun-seq-priv} and Condition~\eqref{eqn:binfun-seq-priv-set}, respectively; again a contradiction.
	This proves the necessity of Condition (iv).

	Note that the claims work only for distributions $\Q_{XY}$ with rational \emph{p.d.f}.
	But for $\Q_{XY}$ with irrational \emph{p.d.f}, we may appeal to continuity of KL divergence to get a distribution $\Q'_{XY}$ with rational \emph{p.d.f} that is arbitrarily close to $\Q_{XY}$.
	This proves the lemma, and hence the theorem.
\end{proof}

\subsection{Proof Sketch for Theorem~\ref{thm:error-exponent} (Achievability)}

\paragraph*{Overview} To show achievability, we first construct a sequence of decision functions $\left(\binfun_A^n, \binfun_B^n\right)_{n \in \mathbb{N}}$ for Alice and Bob, respectively, with the following property. If Alice and Bob observe only the input and output of their corresponding decision functions, then the correctness-privacy error exponent achieved by this sequence of decision functions match the region described by the theorem. We would then compute these functions using perfectly secure protocols. The view of such protocol reveals no more information than the input and output of the computed function.

\paragraph*{Description of $\binfun_A^n$ and $\binfun_B^n$} 
Fix $\Q_{XY} \in \types{\cX \times \cY}{n}$.
For all $(x^n, y^n) \in \typeclass{\cX \times \cY}{n}{\Q_{XY}}$ and $\theta \in \{0, 1\}$, $\binfun_A^n(x^n, y^n) = \theta$ if one of the following condition is satisfied and $\binfun_A^n(x^n, y^n) = 0$ otherwise.
\newpage
\begin{align}
	\kl{\Q_{XY}}{\P^\theta_{XY}} \le \expc,\\
	(\kl{\Q_X}{\P^\theta_X} \le \expc) \wedge (\kl{\Q_X}{\P^{1 - \theta}_X} > \expc),\\
	(\kl{\Q_X}{\P^\theta_X} \le \expc) \wedge (\klc{\Q_{Y|X}}{\P^\theta_{Y|X}}{\Q_X} \le \expp).
\end{align}
Similarly, for $\theta \in \{0, 1\}$, $\binfun_B^n(x^n, y^n) = \theta$ if one of the following condition is satisfied and $\binfun_B^n(x^n, y^n) = 0$ otherwise.
\begin{align}
	\kl{\Q_{XY}}{\P^\theta_{XY}} \le \expc,\\
	(\kl{\Q_Y}{\P^\theta_Y} \le \expc) \wedge (\kl{\Q_Y}{\P^{1 - \theta}_Y} > \expc),\\
	(\kl{\Q_Y}{\P^\theta_Y} \le \expc) \wedge (\klc{\Q_{X|Y}}{\P^\theta_{X|Y}}{\Q_Y} \le \expp).
\end{align}
Note that when $(\expc, \expp)$ satisfy conditions (i)-(v) in Theorem~\ref{thm:error-exponent}, the above functions map each $(x^n, y^n)$ uniquely to either $0$ or $1$, and are hence well defined.

\begin{claim}\label{clm:ach}
	For the sequence of decision functions $\left(\binfun_A^n, \binfun_B^n\right)_{n \in \mathbb{N}}$, $\exists (\delta_n, \mu_n)_{n \in \bbN}$ such that for all $n \in \mathbb{N}$ and all $\theta \in \{0,1\}$,
	\begin{align*}
		\Prob_{\binfun_A^n | \Theta} (1 - \theta | \theta) \le \delta_n,\;
		\Prob_{\binfun_B^n | \Theta} (1 - \theta | \theta) \le \delta_n,\\
		\Prob\left(\tvd{\dist{Y^n | x^n, \theta}}{\dist{Y^n | x^n, \binfun_A^n, \theta}} \ge \mu_n \right) \le \mu_n, \forall x^n,\\
		\Prob\left(\tvd{\dist{X^n | y^n, \theta}}{\dist{X^n | y^n, \binfun_B^n, \theta}} \ge \mu_n \right) \le \mu_n, \forall y^n,
	\end{align*}
	and
	\begin{align}\label{eqn:ach-condns}
		\expc = \lim \limits_{n \rightarrow \infty} -\frac{1}{n} \log{\delta_n},
		\expp = \lim \limits_{n \rightarrow \infty} -\frac{1}{n} \log{\mu_n}.
	\end{align}
\end{claim}
We now appeal to the following theorem, proved in the Appendix, to obtain a sequence of protocols $(\prot_n)_{n \in \mathbb{N}}$ from $\left(\binfun_A^n, \binfun_B^n\right)_{n \in \mathbb{N}}$ such that $\prot_n$ is $(n, \delta_n, \mu_n)$-private.
This proves the achievability.

\begin{defn}\label{def:ot}
	An oblivious transfer (OT) correlation consists of random variables $W_A, W_B$, where $W_A = (R_0, R_1)$ and $W_B = (B, R_B)$, where $R_0, R_1, B$ \emph{i.i.d.} $\sim \text{Bernoulli}(\frac{1}{2})$.
\end{defn}

\begin{thm}\label{mpc-ot}
	Let $\binfun_A, \binfun_B : \cX \times \cY \rightarrow \{0, 1\}$ be a pair of randomized boolean functions. For sufficiently large $k$, when $W_A, W_B$ consists of $k$ copies of OT correlations, there exists a protocol $\prot$, with the following guarantees.
	\begin{align*}
		\dist{\dec_A | X = x, Y = y} \equiv \dist{\binfun_A(x, y)}, \forall x, y,\\
		\dist{\dec_B | X = x, Y = y} \equiv \dist{\binfun_B(x, y)}, \forall x, y.
	\end{align*}
	For $\theta = 0, 1$, when $\mu \ge 0$,
	\begin{multline*}
		\Prob(\tvd{\dist{ Y | x, \theta}}{\dist{Y | \binfun_A, x, \theta}} \ge \mu) \le \mu\\ \implies \Prob(\tvd{\dist{ Y | x, \theta}}{\dist{Y | V_A, x, w_A, \theta}} \ge \mu) \le \mu, \\ \forall x, w_A,
	\end{multline*}
	\begin{multline*}
		\Prob(\tvd{\dist{ X | y, \theta}}{\dist{X | \binfun_B, y, \theta}} \ge \mu) \le \mu\\ \implies \Prob(\tvd{\dist{ X | y, \theta}}{\dist{X | V_B, y, w_B, \theta}} \ge \mu) \le \mu, \\ \forall y, w_B.
	\end{multline*}
\end{thm}


\section*{Acknowledgements}
MM and VP acknowledge useful discussions with Dr. Jithin Ravi, Universidad Carlos III de Madrid, Legan\'es, Spain.
VN and VP were supported by the Department of Atomic Energy, Government of India, under project no. 12-R\&D-TFR-5.01-0500.

\onecolumn
\bibliography{references}
\begin{appendices}
\section{Some Useful Lemmas}

We would use the following Lemmas in the proofs of our results.
This section can be referred to when needed in the main proofs.
\begin{lemma}\label{lem:avg-privacy}
	For the joint distribution described in~\eqref{eqn:prot-dist}, the following statements are true.
	\begin{enumerate}
		\item For $\theta = 0, 1$,
		\begin{multline}\label{eqn:avg-priv1}
			\Prob  \left( \tvd{\dist{Y^n | X^n,\Theta = \theta}}{\dist{Y^n | V_A, \Theta = \theta}} \ge \mu \right)  \le \mu\\
			\implies \sum_{x^n \in \cX^n, v_A \in \cV_A} \Prob_{X^n, V_A | \Theta} (x^n, v_A | \theta) \tvd{\dist{Y^n | x^n, \theta}}{\dist{Y^n | v_A, \theta}} \le 2\mu.
		\end{multline}
		\item For $\theta = 0, 1$, and any $x^n \in \cX^n$,
		\begin{multline}\label{eqn:avg-priv2}
			\Prob  \left( \tvd{\dist{Y^n | x^n,\Theta = \theta}}{\dist{Y^n | V_A, x^n, \Theta = \theta}} \ge \mu \right)  \le \mu\\
			\implies \sum_{v_A \in \cV_A} \Prob_{V_A | X^n, \Theta} (v_A | x^n, \theta) \tvd{\dist{Y^n | x^n, \theta}}{\dist{Y^n | v_A, \theta}} \le 2\mu.
		\end{multline}
	\end{enumerate}
\end{lemma}
\begin{proof}
	Define set $S \subseteq \cX^n \times \cV_A$ such that $(x^n, v_A) \in S$ if and only if
		\begin{align*}
			\tvd{\dist{Y^n | x^n,\Theta = \theta}}{\dist{Y^n | v_A, \Theta = \theta}} \ge \mu.
		\end{align*}
		By the assumption $\Prob(S) \le \mu$.
		We can bound the RHS of the statement~\eqref{eqn:avg-priv1} as follows.
		\begin{multline*}
			\sum_{(x^n, v_A) \in S} \Prob_{X^n, V_A | \Theta} (x^n, v_A | \theta) \tvd{\dist{Y^n | x^n, \theta}}{\dist{Y^n | v_A, \theta}}\\
			+ \sum_{(x^n, v_A) \notin S} \Prob_{X^n, V_A | \Theta} (x^n, v_A | \theta) \tvd{\dist{Y^n | x^n, \theta}}{\dist{Y^n | v_A, \theta}}.
		\end{multline*}
		The first term in the above expression can be bounded by $\mu$ since $\Prob(S) \le \mu$ and statistical distance is upper bounded by 1.
		The second term can be bounded by $\mu$ since the statistical distance in the term is at most $\mu$ by the definition of $S$.
		Statement~\eqref{eqn:avg-priv2} lemma can be proved identically.
\end{proof}
\begin{lemma}\label{lem:kld-prob}
	Let $\Q \in \types{\cX \times \cY}{n}$, for a fixed $x^n$ (of appropriate type),
	\begin{align*}
		2^{-n \left(\klc{Q_{Y|X}}{P^\theta_{Y|X}}{\Q_{X}} + |\cX \times \cY|\frac{\log{2n}}{n}\right)} \le \sum_{y^n : (x^n, y^n) \in \typeclass{\cX^n \times \cY^n}{n}{\Q_{XY}}}\Prob_{Y^n | X^n, \Theta}( y^n | x^n, \theta) \le 2^{-n \left(\klc{Q_{Y|X}}{P^\theta_{Y|X}}{\Q_{X}}\right)}
	\end{align*}
\end{lemma}
\begin{proof}
	For any $(x^n, y^n) : \type{x^n, y^n} = \Q_{XY}$,
	\begin{align*}
		\Prob_{Y^n | X^n, \Theta}( y^n | x^n, \theta) = \prod_{(x, y) \in \cX \times \cY} \left(\Prob_{Y | X, \Theta}( y | x, \theta)\right)^{n \cdot \Q_{XY}(x, y)}
	\end{align*}
	Hence,
	\begin{align}
		\sum_{y^n : \type{(x^n, y^n)} = \Q_{XY}}\Prob_{Y^n | X^n, \Theta}( y^n | x^n, \theta) = |\{y^n : \type{(x^n, y^n)} = \Q_{XY}\}| \cdot \prod_{(x, y) \in \cX \times \cY} \left(\Prob_{Y | X, \Theta}( y | x, \theta)\right)^{n \cdot \Q_{XY}(x, y)}.\label{eqn:cond-type-prob}
	\end{align}
	Before we bound the size of the set $\{y^n : \type{(x^n, y^n)} = \Q_{XY}\}$, we quote the following theorem verbatim from~\cite{CoverT06}.
	\begin{thm}\cite[Theorem 11.1.3]{CoverT06}\label{thm:cover-typeclass}
		For any type class $\Q_X \in \types{\cX}{n}$,
		\begin{align*}
			\frac{1}{(n + 1)^{|\cX|}} 2^{n H(\Q_X)} \le \left|\typeclass{\cX}{n}{\Q_X}\right| \le 2^{n H(\Q_X)}.
		\end{align*}
	\end{thm}
	Observe that,
	\begin{align*}
		|\{y^n : \type{(x^n, y^n)} = \Q_{XY}\}| = \prod_{x \in \cX} \left|\left\{y^{n \cdot \Q_X(x)} \in \cY^{n \cdot \Q_X(x)} : \type{y^{n \cdot \Q_X(x)}} = \Q_{Y|X = x}\right\}\right|.
	\end{align*}
	By using Theorem~\ref{thm:cover-typeclass}, we may bound this as
	\begin{align*}
		\prod_{x \in \cX} \frac{1}{(n + 1)^{|\cX|}}2^{n \cdot \Q_X(x) \cdot H(\Q_{Y | X = x})} \le |\{y^n : \type{(x^n, y^n)} = \Q_{XY}\}| \le \prod_{x \in \cX} 2^{n \cdot \Q_X(x) \cdot H(\Q_{Y | X = x})}.
	\end{align*}
	Using the above observation and the equality~\eqref{eqn:cond-type-prob}, we get the following lower bound,
	\begin{align*}
		&\log{\sum_{y^n : \type{(x^n, y^n)} = \Q_{XY}}\Prob_{Y^n | X^n, \Theta}( y^n | x^n, \theta)}\\ 
		\ge & \sum_{x \in \cX} -n \cdot \Q_X(x) \left(\sum_{y \in \cY} Q_{Y|X}(y|x) \log{Q_{Y|X}(y|x)} + \frac{|\cY|\log{2n}}{n}\right) + \sum_{(x, y) \in \cX \times \cY} n \cdot \Q_{XY}(x, y) \log{\P^{\theta}_{Y|X}(y|x)}\\
		= & -n \left(\sum_{(x, y) \in \cX \times \cY} \Q_{XY}(x, y) \log{\frac{\Q_{Y|X}(y|x)}{\P^\theta_{Y|X}(y|x)}} + \frac{|\cX \times \cY|\log{2n}}{n}\right) = -n \left(\klc{Q_{Y|X}}{P^\theta_{Y|X}}{\Q_{X}} + |\cX \times \cY|\frac{\log{2n}}{n}\right).
	\end{align*}
	Using the above observation and the equality~\eqref{eqn:cond-type-prob}, we get the following upper bound,
	\begin{align*}
		&\log{\sum_{y^n : \type{(x^n, y^n)} = \Q_{XY}}\Prob_{Y^n | X^n, \Theta}( y^n | x^n, \theta)}\\ 
		\le & \sum_{x \in \cX} -n \cdot \Q_X(x) \left(\sum_{y \in \cY} Q_{Y|X}(y|x) \log{Q_{Y|X}(y|x)}\right) + \sum_{(x, y) \in \cX \times \cY} n \cdot \Q_{XY}(x, y) \log{\P^{\theta}_{Y|X}(y|x)}\\
		= & -n \left(\sum_{(x, y) \in \cX \times \cY} \Q_{XY}(x, y) \log{\frac{\Q_{Y|X}(y|x)}{\P^\theta_{Y|X}(y|x)}}\right) = -n \left(\klc{Q_{Y|X}}{P^\theta_{Y|X}}{\Q_{X}}\right).
	\end{align*}
	This proves the claim.
\end{proof}

\section{Example: Missing Proofs}
\begin{proof}[Proof of Claim~\ref{clm:eg-priv}]
	The statistical distance given in the claim may be expanded as follows.
	\begin{align*}
		& \tvd{\dist{Z^n, V_A^{\prot} (\hat{X}^n, Z^n)}}{\dist{Z^n, V_A^{\prot}(\hat{X}^n, \hat{Y}^n)}}\\
		& = \tvd{\dist{Z^n, \hat{X}^n, V_A^{\prot} (\hat{X}^n, Z^n)}}{\dist{Z^n, \hat{X}^n, V_A^{\prot}(\hat{X}^n, \hat{Y}^n)}}\\
		& = \frac{1}{2}\sum_{(z^n, x^n, v_A) \in \cW^n \times \cX^n \times \cV_A} \left|\Prob_{Z^n, \hat{X}^n, V_A^{\prot}(\hat{X}^n, Z^n)} (z^n, x^n, v_A) - \Prob_{Z^n, \hat{X}^n, V_A^{\prot}(\hat{X}^n, \hat{Y}^n)} (z^n, x^n, v_A)\right|\\
		& \stackrel{(a)}{=} \; \frac{1}{2}\sum_{(z^n, x^n, v_A) \in \cW^n \times \cX^n \times \cV_A} \left|\Prob_{Z^n, \hat{X}^n} (z^n, x^n) \cdot \left(\Prob_{V_A^{\prot}(\hat{X}^n, Z^n)| \hat{X}^n, Z^n}(v_A | x^n, z^n) - \Prob_{V_A^{\prot}(\hat{X}^n, \hat{Y}^n)| \hat{X}^n}(v_A | x^n) \right)\right|\\
		& = \; \frac{1}{2}\sum_{(x^n, z^n) \in \cX^n \times \cW^n} \Prob(x^n, z^n) \cdot \sum_{v_A \in \cV_A}\left|\Prob_{V_A^{\prot}(\hat{X}^n, Z^n)| \hat{X}^n, Z^n}(v_A | x^n, z^n) - \Prob_{V_A^{\prot}(\hat{X}^n, \hat{Y}^n)| \hat{X}^n}(v_A | x^n) \right|\\
		& \stackrel{(b)}{=} \; \frac{1}{2}\sum_{(x^n, z^n) \in \cX^n \times \cW^n} \Prob(x^n, z^n) \cdot \sum_{v_A \in \cV_A}\left| \text{Term 1} - \text{Term 2}\right|
		\intertext{where, }
		\text{Term 1 } \MM{=} & \; \frac{\Prob_{ V_A^{\prot}(\hat{X}^n, Z^n) | \hat{X}^n}(v_A | x^n) \cdot \Prob_{Z^n| V_A^{\prot}(\hat{X}^n, Z^n), \hat{X}^n}(z^n | v_A, x^n)}{\Prob_{Z^n | \hat{X}^n} (z^n | x^n)}\\
		\text{Term 2 } \MM{=} & \; \frac{\Prob_{V_A^{\prot}(\hat{X}^n, \hat{Y}^n)| \hat{X}^n}(v_A | x^n) \cdot \Prob_{Z^n | \hat{X}^n} (z^n | x^n)}{\Prob_{Z^n | \hat{X}^n} (z^n | x^n)}.
	\end{align*}
	In (a), we used the independence of $\hat{Y}^n$ and $Z^n$ and (b) expands the conditional probability using Bayes's Theorem.
	Since $(\hat{X}^n, \hat{Y}^n)$ and $(\hat{X}^n, Z^n)$ are identically distributed, for all $x^n, v_a$,
	\begin{align*}
		\Prob_{ V_A^{\prot}(\hat{X}^n, Z^n) | \hat{X}^n}(v_A | x^n) = \Prob_{ V_A^{\prot}(\hat{X}^n, \hat{Y}^n) | \hat{X}^n}(v_A | x^n) \text{ for all } x^n \in \cX^n, v_A \in \cV_A.
	\end{align*}
	Using this, the above expression can be simplified as,
	\begin{align*}
		& \frac{1}{2}\sum_{(x^n, z^n) \in \cX^n \times \cW^n} \Prob(x^n, z^n) \cdot \sum_{v_A \in \cV_A} \frac{\Prob_{ V_A^{\prot}(\hat{X}^n, Z^n) | \hat{X}^n}(v_A | x^n)}{\Prob_{Z^n | \hat{X}^n} (z^n | x^n)}\left|\Prob_{Z^n| V_A^{\prot}(\hat{X}^n, Z^n), \hat{X}^n}(z^n | v_A, x^n) - \Prob_{Z^n | \hat{X}^n}(z^n | x^n) \right|\\
		= & \frac{1}{2}\sum_{x^n \in \cX^n} \Prob(x^n) \cdot \sum_{v^A \in \cV_A} \Prob_{V_A^{\prot}(\hat{X}^n, Z^n) | \hat{X}^n} (v_A | x^n) \cdot \sum_{z^n \in \cW^n} \left|\Prob_{Z^n| V_A^{\prot}(\hat{X}^n, Z^n), \hat{X}^n}(z^n | v_A, x^n) - \Prob_{Z^n | \hat{X}^n}(z^n | x^n) \right|\\
		\stackrel{(a)}{=} & \frac{1}{2}\sum_{(x^n, v_A) \in \cX^n \times \cV_A} \Prob_{\hat{X}^n, V_A | \Theta = 0} (x^n, v_A) \cdot \sum_{z^n \in \cW^n} \left|\Prob_{Z^n| V_A^{\prot}(\hat{X}^n, Z^n), \hat{X}^n}(z^n | v_A, x^n) - \Prob_{Z^n | \hat{X}^n}(z^n | x^n) \right|\\
		\stackrel{(b)}{=} & \sum_{(x^n, v_A) \in \cX^n \times \cV_A} \Prob_{\hat{X}^n, V_A | \Theta = 0} (x^n, v_A) \cdot \tvd{\dist{z^n | v_A, \Theta = 0}} {\dist{z^n | x^n, \Theta = 0}} \stackrel{(c)}{\le} 2 \mu.
	\end{align*}
	In (a), we used the equality $\Prob(x^n) \cdot \Prob_{V_A^{\prot}(\hat{X}^n, Z^n) | \hat{X}^n} (v_A | x^n) = \Prob_{V_A^{\prot}(\hat{X}^n, Z^n), \hat{X}^n} (v_A, x^n)$.
	This follows from the description of the joint distribution~\eqref{eqn:prot-dist} and the fact that $\Theta = 0$ since $\hat{X}^n$ and $Z^n$ are independent.
	(b) uses the definition of statistical distance and (c) follows from Lemma~\ref{lem:avg-privacy}.

\end{proof}

\section{Converse: Missing Proofs}

\begin{proof}[Proof of Lemma~\ref{lem:binfun-type}]
	The correctness condition (i) follows directly from the definition of $\binfun_A$ and $\binfun_B$ and $\delta$-correctness of $\prot$.
	We prove statement (ii) by an averaging argument.
	Since $\dec_A = \decfun_A(V_A)$ (resp. $\dec_B = \decfun_B(V_B)$) is a deterministic function of $V_A$ (resp. $V_B$), and since $\Prob_{\binfun_A | X^n, Y^n} (i | x^n, y^n) = \Prob_{\dec_A | X^n, Y^n}(i | x^n, y^n)$ for all $x^n, y^n$,
	\begin{align*}
	\Prob_{Y^n|\binfun_A, X^n, \Theta}(y^n | i, x^n, \theta) &= \Prob_{Y^n | \dec_A, X^n, \Theta}(y^n | i, x^n, \theta) = \sum_{v_A \in \cV_A}\Prob_{Y^n, V_A | \dec_A, X^n, \Theta}(y^n, v_A | i, x^n, \theta)\\
	& = \sum_{v_A \in \cV_A}\Prob_{V_A | \dec_A, X^n, \Theta}(v_A | i, x^n, \theta) \cdot \Prob_{Y^n | V_A, \dec_A, X^n, \Theta}(y^n | v_A, i, x^n, \theta)\\
	& \stackrel{(a)}{=} \sum_{v_A: \decfun_A(v_A) = i}\frac{\Prob_{V_A | X^n, \Theta}(v_A | x^n, \theta)}{\Prob_{\dec_A | X^n, \Theta} (i | x^n, \theta)} \cdot \Prob_{Y^n | V_A, \Theta}(y^n | v_A, \theta)\\
	\end{align*}
	Here, (a) crucially uses the fact that $x^n, i$ are part of the view $v_A$. Using the above observation we proceed as follows,
	\begin{align*}
		& \sum_{i = 0}^1\Prob_{\binfun_A | X^n, \Theta} (i | x^n, \theta) \cdot \left(\tvd{\dist{Y^n|x^n, \theta}}{\dist{Y^n|\binfun_A = i, x^n, \theta}}\right)\\
		= & \sum_{i = 0}^1 \Prob_{\dec_A | X^n, \Theta} (i | x^n, \theta) \cdot \frac{1}{2} \sum_{y^n \in \cY^n} \left| \Prob_{Y^n | X^n, \Theta}(y^n | x^n, \theta) - \Prob_{Y^n|\binfun_A, X^n, \Theta}(y^n | i, x^n, \theta)\right|\\
		= & \sum_{i = 0}^1 \Prob_{\dec_A | X^n, \Theta} (i | x^n, \theta) \cdot \frac{1}{2} \sum_{y^n \in \cY^n} \left| \Prob_{Y^n | X^n, \Theta}(y^n | x^n, \theta) - \sum_{v_A: \decfun_A(v_A) = i}\frac{\Prob_{V_A | X^n, \Theta}(v_A | x^n, \theta)}{\Prob_{\dec_A | X^n, \Theta} (i | x^n, \theta)} \cdot \Prob_{Y^n | V_A, \Theta}(y^n | v_A, \theta)\right|\\
		\stackrel{(a)}{\le} & \sum_{i = 0}^1 \Prob_{\dec_A | X^n, \Theta} (i | x^n, \theta) \cdot \frac{1}{2} \sum_{y^n \in \cY^n} \sum_{v_A: \decfun_A(v_A) = i}\frac{\Prob_{V_A | X^n, \Theta}(v_A | x^n, \theta)}{\Prob_{\dec_A | X^n, \Theta} (i | x^n, \theta)} \cdot \left| \Prob_{Y^n | X^n, \Theta}(y^n | x^n, \theta) - \Prob_{Y^n | V_A, \Theta}(y^n | v_A, \theta)\right|\\
		= & \sum_{i = 0}^1 \sum_{v_A: \decfun_A(v_A) = i}\Prob_{V_A | X^n, \Theta}(v_A | x^n, \theta) \cdot \frac{1}{2} \sum_{y^n \in \cY^n} \cdot \left| \Prob_{Y^n | X^n, \Theta}(y^n | x^n, \theta) - \Prob_{Y^n | V_A, \Theta}(y^n | v_A, \theta)\right|\\
		= & \sum_{v_A \in \cV}\Prob_{V_A | X^n, \Theta}(v_A | x^n, \theta) \cdot \frac{1}{2} \sum_{y^n \in \cY^n} \cdot \left| \Prob_{Y^n | X^n, \Theta}(y^n | x^n, \theta) - \Prob_{Y^n | V_A, \Theta}(y^n | v_A, \theta)\right|\\
		= & \sum_{v_A \in \cV}\Prob_{V_A | X^n, \Theta}(v_A | x^n, \theta) \cdot \tvd{\dist{Y^n | x^n, \theta}}{\dist{Y^n | v_A, \theta}} \stackrel{(b)}{\le} 2\mu.
	\end{align*}
	Here, (a) follows from Jensen's inequality for absolute value function and (b) follows from Lemma~\ref{lem:avg-privacy}.
	This proves the lemma.
\end{proof}
\begin{proof}[Proof of Lemma~\ref{lem:binfun-seq}]
	We will first prove the following two claims.
	\newtheorem*{clm:binfun-seq-corr}{Claim \ref{clm:binfun-seq-corr}}
	\begin{clm:binfun-seq-corr}
		For $\theta = 0, 1$, for large enough $n$ in the sequence $(n_i)_{i \in \mathbb{N}}$, if $\Q_{XY} \in \types{\cX \times \cY}{n}$ and $\kl{\Q_{XY}}{\P^\theta_{XY}} < \expc$, $\Prob(X^n \in S | X^n \in \typeclass{\cX}{n}{\Q_X}, \Theta = \theta) \ge \frac{99}{100}$, where, 
		\begin{align*}
			S = \{x^n \in \typeclass{\cX}{n}{\Q_X} : \frac{\sum_{y^n : (x^n, y^n) \in \typeclass{\cX \times \cY}{n}{\Q_{XY}}} \Prob(\binfun_A = \theta | x^n, y^n)}{|\{y^n : (x^n, y^n) \in \typeclass{\cX \times \cY}{n}{\Q_{XY}}\}|} \ge \frac{99}{100}\}.
		\end{align*}
	\end{clm:binfun-seq-corr}
	\begin{proof}
		For $\tau > 0$, let $\kl{\Q_{XY}}{\P^\theta_{XY}} = \expc - \tau$. For some $\tau' < \tau$, choose $n$ from the sequence $(n_i)_{i \in \mathbb{N}}$ such that $\expc - \tau' \le -\frac{1}{n} \log{\delta_n}$.
		Hence, we have, $\Prob_{\binfun_A^n | \Theta} (1 - \theta | \theta) \le 2^{-n \cdot (\expc - \tau')}$.
		\begin{align*}
			2^{-n \cdot (\expc - \tau')} & \ge \Prob_{\binfun_A^n | \Theta} (1 - \theta | \theta)\\
			& \ge \Prob\left((X^n, Y^n) \in \typeclass{\cX \times \cY}{n}{\Q_{XY}} | \Theta = \theta\right) \Prob\left(\binfun_A^n = 1 - \theta | (X^n, Y^n) \in \typeclass{\cX \times \cY}{n}{\Q_{XY}}, \Theta = \theta\right)
		\end{align*}
		We expand the last expression as,
		\begin{multline*}
			\Prob\left((X^n, Y^n) \in \typeclass{\cX \times \cY}{n}{\Q_{XY}} | \Theta = \theta\right) \\\sum_{(x^n, y^n) \in \typeclass{\cX \times \cY}{n}{\Q_{XY}}} \Prob\left((X^n, Y^n) = (x^n, y^n) | (X^n, Y^n) \in \typeclass{\cX \times \cY}{n}{\Q_{XY}}, \Theta = \theta\right)\Prob\left(\binfun_A^n = 1 - \theta | (X^n, Y^n) = (x^n, y^n), \Theta = \theta\right).
		\end{multline*}
		In the above expression, since the probability of every member of a typeclass is the same irrespective of the hypothesis, $\Prob\left((X^n, Y^n) = (x^n, y^n) | (X^n, Y^n) \in \typeclass{\cX \times \cY}{n}{\Q_{XY}}, \Theta = \theta\right) = \frac{1}{\left|\typeclass{\cX \times \cY}{n}{\Q_{XY}}\right|}$. 
		Additionally, $\binfun_A$ is independent of $\Theta$ conditioned on $(x^n, y^n)$.
		Hence, from the above observations we get the following inequality.
		\begin{align*}
			2^{-n \cdot (\expc - \tau')} \ge \Prob\left((X^n, Y^n) \in \typeclass{\cX \times \cY}{n}{\Q_{XY}} | \Theta = \theta\right) \cdot
			\frac{1}{\left|\typeclass{\cX \times \cY}{n}{\Q_{XY}}\right|} \sum_{(x^n, y^n) \in \typeclass{\cX \times \cY}{n}{\Q_{XY}}} \Prob\left(\binfun_A^n = 1 - \theta | (X^n, Y^n) = (x^n, y^n)\right).
		\end{align*}
		By Theorem 11.1.4 in \cite{CoverT06}, we have,
		\begin{align*}
			\Prob\left((X^n, Y^n) \in \typeclass{\cX \times \cY}{n}{\Q_{XY}} | \Theta = \theta\right) \ge \frac{1}{(n + 1)^{|\cX \times \cY|}} 2^{-n \cdot \kl{\Q_{XY}}{\P^\theta_{XY}}} = \frac{1}{(n + 1)^{|\cX \times \cY|}} 2^{-n \cdot (\expc - \tau)}.
		\end{align*}
		Hence,
		\begin{multline*}
			\frac{1}{\left|\typeclass{\cX \times \cY}{n}{\Q_{XY}}\right|} \sum_{(x^n, y^n) \in \typeclass{\cX \times \cY}{n}{\Q_{XY}}} \Prob\left(\binfun_A^n = 1 - \theta | (X^n, Y^n) = (x^n, y^n)\right)\\ \le \frac{1}{(n + 1)^{|\cX \times \cY|}}2^{-n \cdot (\expc - \tau')} \cdot 2^{n \cdot (\expc - \tau)} = \frac{1}{(n + 1)^{|\cX \times \cY|}}2^{-n \cdot (\tau - \tau')}.
		\end{multline*}
		Choose $n$ large enough to guarantee, $\frac{1}{(n + 1)^{|\cX \times \cY|}}2^{-n \cdot (\tau - \tau')} \le \frac{1}{10^4}$.
		Towards a contradiction, suppose $\Prob(X^n \in S | X^n \in \Q_X, \Theta = \theta) < \frac{99}{100}$.
		Then,
		\begin{align*}
			& \frac{1}{\left|\typeclass{\cX \times \cY}{n}{\Q_{XY}}\right|} \sum_{(x^n, y^n) \in \typeclass{\cX \times \cY}{n}{\Q_{XY}}} \Prob\left(\binfun_A^n = 1 - \theta | (X^n, Y^n) = (x^n, y^n)\right)\\
			\ge & \frac{1}{\left|\typeclass{\cX \times \cY}{n}{\Q_{XY}}\right|} \sum_{x^n \notin S} \frac{\sum_{y^n : (x^n, y^n) \in \typeclass{\cX \times \cY}{n}{\Q_{XY}}} \Prob(\binfun_A = 1 - \theta | x^n, y^n)}{|\{y^n : (x^n, y^n) \in \typeclass{\cX \times \cY}{n}{\Q_{XY}}\}|} \cdot |\{y^n : (x^n, y^n) \in \typeclass{\cX \times \cY}{n}{\Q_{XY}}\}|\\
			> & \frac{1}{100} \frac{\sum_{x^n \notin S}|\{y^n : (x^n, y^n) \in \typeclass{\cX \times \cY}{n}{\Q_{XY}}\}|}{\left|\typeclass{\cX \times \cY}{n}{\Q_{XY}}\right|}
			= \frac{1}{100} \frac{\sum_{x^n \notin S}|\{y^n : (x^n, y^n) \in \typeclass{\cX \times \cY}{n}{\Q_{XY}}\}|}{\sum_{x^n \in \typeclass{\cX}{n}{\Q_X}}|\{y^n : (x^n, y^n) \in \typeclass{\cX \times \cY}{n}{\Q_{XY}}\}|}\\
			& \stackrel{(a)}{=} \frac{1}{100} \frac{|S|}{|\typeclass{\cX}{n}{\Q_X}|} \stackrel{(b)}{=} \frac{1}{100} \Prob(X^n \notin S | X^n \in \typeclass{\cX}{n}{\Q_X}, \Theta = \theta) > \frac{1}{10^4}.
		\end{align*}
		In (a), we crucially use the fact that the size of the set $\{y^n : (x^n, y^n) \in \typeclass{\cX \times \cY}{n}{\Q_{XY}}\}$ is the same for all $x^n \in \typeclass{\cX}{n}{\Q_X}$. In (b), we use the fact that $\frac{|S|}{|\typeclass{\cX}{n}{\Q_X}|} = \Prob(X^n \notin S | X^n \in \typeclass{\cX}{n}{\Q_X}, \Theta = \theta)$ irrespective of the value of $\theta$.
		The above contradiction proves the claim.
	\end{proof}
	\newtheorem*{clm:binfun-seq-priv}{Claim \ref{clm:binfun-seq-priv}}
	\begin{clm:binfun-seq-priv}
		For $\theta = 0, 1$, for all large values $n$ in the sequence $(n_i)_{i \in \mathbb{N}}$, if $\Q_{XY} \in \types{\cX \times \cY}{n}, \kl{\Q_X}{\P^\theta_X} < \expc$, and $\klc{\Q_{Y|X}}{\P^\theta_{Y|X}}{\Q(x)} < \expp$ then, $\Prob(X^n \in S | X^n \in \typeclass{\cX}{n}{\Q_X}, \Theta = \theta) \ge \frac{99}{100}$, where, 
		\begin{align*}
			S = \{x^n \in \typeclass{\cX}{n}{\Q_X} : \frac{\sum_{y^n : (x^n, y^n) \in \typeclass{\cX \times \cY}{n}{\Q_{XY}}} \Prob(\binfun_A = \theta | x^n, y^n)}{|\{y^n : (x^n, y^n) \in \typeclass{\cX \times \cY}{n}{\Q_{XY}}\}|} \ge \frac{80}{100}\}.
		\end{align*}
	\end{clm:binfun-seq-priv}
	\begin{proof}
		For $\tau > 0$, let $\klc{\Q_{Y|X}}{\P^\theta_{Y|X}}{\Q(x)} = \expp - \tau$. For some $\tau' < \tau$, choose $n$ from the sequence $(n_i)_{i \in \mathbb{N}}$ such that $\expp - \tau' \le -\frac{1}{n} \log{\mu_n}$.
		Note that, $\kl{\Q_X \cdot \P^\theta_{Y | X}}{\P^\theta_{XY}} = \kl{\Q_X}{\P^\theta_X} < \expc$.
		We would also require $n$ to be large enough that by to Claim~\ref{clm:binfun-seq-corr}, $\Prob(X^n \in R | X^n \in \typeclass{\cX}{n}{\Q_X}, \Theta = \theta) \ge \frac{99}{100}$, where, 
		\begin{align}
			R = \{x^n \in \typeclass{\cX}{n}{\Q_X} : \frac{\sum_{y^n : (x^n, y^n) \in \typeclass{\cX \times \cY}{n}{\Q_{X} \cdot \P^{\theta}_{Y | X}}} \Prob(\binfun_A^n = \theta | x^n, y^n)}{|\{y^n : (x^n, y^n) \in \typeclass{\cX \times \cY}{n}{\Q_{X} \cdot \P^{\theta}_{Y | X}}\}|} \ge \frac{99}{100}\}.\label{eqn:good-x}
		\end{align}
		Consider $x^n \in R$, we will first establish a lower bound of $\frac{95}{100}$ on $\Prob_{\binfun_A^n | X^n, \Theta}(\theta | x^n, \theta)$.
		Towards a contradiction, suppose $\Prob_{\binfun_A^n | X^n, \Theta}(\theta | x^n, \theta) < \frac{95}{100}$.
		We note that the value of $\Prob_{Y^n | X^n, \Theta}(y^n |x^n, \theta)$ is the same for all $(x^n, y^n) \in \typeclass{\cX \times \cY}{n}{\Q_{X} \cdot \P_{Y|X}}$.
		Let $t$ denote this value.
		Then,
		\begin{align}
			& \sum_{y^n : (x^n, y^n) \in \typeclass{\cX \times \cY}{n}{\Q_{X} \cdot \P_{Y|X}}} \Prob_{Y^n | \binfun_A^n, X^n, \Theta}(y^n | 1 - \theta, x^n, \theta) \nonumber\\
			= & \sum_{y^n : (x^n, y^n) \in \typeclass{\cX \times \cY}{n}{\Q_{X} \cdot \P_{Y|X}}} \frac{\Prob_{Y^n | X^n, \Theta}(y^n |x^n, \theta) \cdot \Prob_{\binfun_A^n | X^n, Y^n, \Theta} (1 - \theta | x^n, y^n, \theta)}{\Prob_{\binfun_A^n | X^n, \Theta}(1 - \theta | x^n, \theta)} \nonumber\\
			\stackrel{(a)}{=} &t \cdot |\{y^n : (x^n, y^n) \in \typeclass{\cX \times \cY}{n}{\Q_{X} \cdot \P^{\theta}_{Y | X}}\}| \\
			& \qquad \sum_{y^n : (x^n, y^n) \in \typeclass{\cX \times \cY}{n}{\Q_{X} \cdot \P_{Y|X}}} \frac{\Prob_{\binfun_A^n | X^n, Y^n, \Theta} (1 - \theta | x^n, y^n, \theta)}{\Prob_{\binfun_A^n | X^n, \Theta}(1 - \theta | x^n, \theta) \cdot |\{y^n : (x^n, y^n) \in \typeclass{\cX \times \cY}{n}{\Q_{X} \cdot \P^{\theta}_{Y | X}}\}|} \nonumber\\
			\stackrel{(b)}{=} & \sum_{y^n : (x^n, y^n) \in \typeclass{\cX \times \cY}{n}{\Q_{X} \cdot \P_{Y|X}}} \Prob_{Y^n | X^n, \Theta}(y^n |x^n, \theta) \frac{\sum_{y^n : (x^n, y^n) \in \typeclass{\cX \times \cY}{n}{\Q_{X} \cdot \P_{Y|X}}}\Prob_{\binfun_A^n | X^n, Y^n, \Theta} (1 - \theta | x^n, y^n, \theta)}{\Prob_{\binfun_A^n | X^n, \Theta}(1 - \theta | x^n, \theta) \cdot |\{y^n : (x^n, y^n) \in \typeclass{\cX \times \cY}{n}{\Q_{X} \cdot \P^{\theta}_{Y | X}}\}|} \nonumber\\
			\stackrel{(c)}{\le} & \frac{\frac{1}{100}}{\frac{5}{100}} \sum_{y^n : (x^n, y^n) \in \typeclass{\cX \times \cY}{n}{\Q_{X} \cdot \P_{Y|X}}} \Prob_{Y^n | X^n, \Theta}(y^n |x^n, \theta) = \frac{1}{5} \sum_{y^n : (x^n, y^n) \in \typeclass{\cX \times \cY}{n}{\Q_{X} \cdot \P_{Y|X}}} \Prob_{Y^n | X^n, \Theta}(y^n |x^n, \theta).\label{eqn:bound-alpha-converse}
		\end{align}
		Equalities (a) and (b) follow from the definition of $t$.
		In (c), we use the fact that $x^n \in R$ (see definition of set $R$ in~\eqref{eqn:good-x}), and our assumption that $\Prob_{\binfun_A^n | X^n, \Theta}(\theta | x^n, \theta) \le \frac{95}{100}$.
		By the privacy condition (ii), we have,
		\begin{align*}
			2 \mu_n & \ge \sum_{i \in \{0, 1\}} \Prob_{\binfun_A^n | X^n, \Theta} (i | x^n, \theta) \cdot \tvd{\dist{Y^n|x^n, \theta}}{\dist{Y^n|\binfun_A^n = i, x^n, \theta}}\\
			& \ge \Prob_{\binfun_A^n | X^n, \Theta} (1 - \theta | x^n, \theta) \cdot \tvd{\dist{Y^n|x^n, \theta}}{\dist{Y^n|\binfun_A^n = 1 - \theta, x^n, \theta}}\\
			& \ge \frac{5}{100} \cdot \frac{1}{2} \sum_{y^n \in \cY^n} \left| \Prob_{Y^n | X^n, \Theta}(y^n | x^n, \theta) - \Prob_{Y^n | \binfun_A^n, X^n, \Theta}(y^n | 1 - \theta, x^n, \theta)\right|\\
			& \ge \frac{5}{200} \sum_{y^n : (x^n, y^n) \in \typeclass{\cX \times \cY}{n}{\Q_{X} \cdot \P_{Y|X}}} \left| \Prob_{Y^n | X^n, \Theta}(y^n | x^n, \theta) - \Prob_{Y^n | \binfun_A^n, X^n, \Theta}(y^n | 1 - \theta, x^n, \theta)\right|\\
			& \stackrel{(a)}{=} \frac{5}{200} \cdot t \cdot \sum_{y^n : (x^n, y^n) \in \typeclass{\cX \times \cY}{n}{\Q_{X} \cdot \P_{Y|X}}} \left| 1 - \frac{\Prob_{Y^n | \binfun_A^n, X^n, \Theta}(y^n | 1 - \theta, x^n, \theta)}{t}\right|\\
			& \stackrel{(b)}{\ge} \frac{5}{200} \cdot t \cdot \left| \sum_{y^n : (x^n, y^n) \in \typeclass{\cX \times \cY}{n}{\Q_{X} \cdot \P_{Y|X}}} \left(1 - \frac{\Prob_{Y^n | \binfun_A^n, X^n, \Theta}(y^n | 1 - \theta, x^n, \theta)}{t}\right)\right|\\
			& = \frac{5}{200} \sum_{y^n : (x^n, y^n) \in \typeclass{\cX \times \cY}{n}{\Q_{X} \cdot \P_{Y|X}}} \Prob_{Y^n | X^n, \Theta}(y^n |x^n, \theta) - \sum_{y^n : (x^n, y^n) \in \typeclass{\cX \times \cY}{n}{\Q_{X} \cdot \P_{Y|X}}} \Prob_{Y^n | \binfun_A^n, X^n, \Theta}(y^n | 1 - \theta, x^n, \theta)\\
			& \stackrel{(c)}{\ge} \frac{5}{200} \sum_{y^n : (x^n, y^n) \in \typeclass{\cX \times \cY}{n}{\Q_{X} \cdot \P_{Y|X}}} \left(1 - \frac{1}{5}\right) \Prob_{Y^n | X^n, \Theta}(y^n |x^n, \theta)\\
			& \ge \frac{1}{10^4} \sum_{y^n : (x^n, y^n) \in \typeclass{\cX \times \cY}{n}{\Q_{X} \cdot \P_{Y|X}}} \Prob_{Y^n | X^n, \Theta}(y^n |x^n, \theta) \stackrel{(d)}{\ge} \frac{1}{10^4} 2^{-n\left(\klc{\P^\theta_{Y|X}}{\P^\theta_{Y|X}}{\Q_X} + |\cX \times \cY|\frac{\log{2n}}{n}\right)},
		\end{align*}
		In (a) we use the the fact that $\Prob_{Y^n | X^n, \Theta}(y^n |x^n, \theta)$ is the same for all $(x^n, y^n) \in \typeclass{\cX \times \cY}{n}{\Q_{X} \cdot \P_{Y|X}}$, which we have respresented by $t$.
		(b) follows from Jensen's inequality, (c) from~\eqref{eqn:bound-alpha-converse}, and, finally, (d) follows from Lemma~\ref{lem:kld-prob}.
		Substituting for $2 \mu_n$ and using the using the above bound, we get,
		\begin{align*}
			& 2 \cdot 2^{-n \cdot (\expp - \tau')} \ge 2 \mu_n \ge \frac{1}{10^4} 2^{-n \left(|\cX \times \cY|\frac{\log{2n}}{n}\right)}.
		\end{align*}
		For large values of $n$, this is a contradiction. Thus we have established that for all $x^n \in R$, $\Prob_{\binfun_A^n | X^n, \Theta}(\theta | x^n, \theta) \ge \frac{95}{100}$.

		For $x^n \in R$, suppose 
		\begin{align}
			\frac{\sum_{y^n : (x^n, y^n) \in \typeclass{\cX \times \cY}{n}{\Q_{XY}}} \Prob(\binfun_A = \theta | x^n, y^n)}{|\{y^n : (x^n, y^n) \in \typeclass{\cX \times \cY}{n}{\Q_{XY}}\}|} < \frac{80}{100}.\label{eqn:conv-clm2-assumption}
		\end{align}
		Using an argument similar to the one used in showing~\eqref{eqn:bound-alpha-converse},
		\begin{align}
			& \sum_{y^n : (x^n, y^n) \in \typeclass{\cX \times \cY}{n}{\Q_{XY}}} \Prob_{Y^n | \binfun_A^n, X^n, \Theta}(y^n |\theta, x^n, \theta) \nonumber\\
			= & \sum_{y^n : (x^n, y^n) \in \typeclass{\cX \times \cY}{n}{\Q_{XY}}} \Prob_{Y^n | X^n, \Theta}(y^n |x^n, \theta) \frac{\sum_{y^n : (x^n, y^n) \in \typeclass{\cX \times \cY}{n}{\Q_{XY}}}\Prob_{\binfun_A^n | X^n, Y^n, \Theta} (\theta | x^n, y^n, \theta)}{\Prob_{\binfun_A^n | X^n, \Theta}(\theta | x^n, \theta) \cdot |\{y^n : (x^n, y^n) \in \typeclass{\cX \times \cY}{n}{\Q_{XY}}\}|} \nonumber\\
			\stackrel{(a)}{\le} & \frac{80}{95}\sum_{y^n : (x^n, y^n) \in \typeclass{\cX \times \cY}{n}{\Q_{XY}}} \Prob_{Y^n | X^n, \Theta}(y^n |x^n, \theta).
		\end{align}
		In (a), we used the bound $\Prob_{\binfun_A^n | X^n, \Theta}(\theta | x^n, \theta) \ge \frac{95}{100}$ and our assumption~\eqref{eqn:conv-clm2-assumption}.

		By the privacy condition (ii), we have,
		\begin{align*}
			2 \mu_n & \ge \sum_{i \in \{0, 1\}} \Prob_{\binfun_A^n | X^n, \Theta} (i | x^n, \theta) \cdot \tvd{\dist{Y^n|x^n, \theta}}{\dist{Y^n|\binfun_A^n = i, x^n, \theta}}\\
			& \ge \Prob_{\binfun_A^n | X^n, \Theta} (\theta | x^n, \theta) \cdot \tvd{\dist{Y^n|x^n, \theta}}{\dist{Y^n|\binfun_A^n = \theta, x^n, \theta}}\\
			& \ge \frac{95}{100} \cdot \frac{1}{2}\sum_{y^n \in \cY^n} \left| \Prob_{Y^n | X^n, \Theta}(y^n | x^n, \theta) - \Prob_{Y^n | \binfun_A^n, X^n, \Theta}(y^n | \binfun_A^n = \theta, x^n, \theta)\right|\\
			& \ge \frac{95}{200} \sum_{y^n : (x^n, y^n) \in \typeclass{\cX \times \cY}{n}{\Q_{XY}}} \left| \Prob_{Y^n | X^n, \Theta}(y^n | x^n, \theta) - \Prob_{Y^n | \binfun_A^n, X^n, \Theta}(y^n | \theta, x^n, \theta)\right|\\
			& \ge \frac{95}{100} \sum_{y^n : (x^n, y^n) \in \typeclass{\cX \times \cY}{n}{\Q_{X} \cdot \P_{Y|X}}} \left(1 - \frac{80}{95}\right) \Prob_{Y^n | X^n, \Theta}(y^n |x^n, \theta)\\
			& \ge \frac{1}{10^4} \sum_{y^n : (x^n, y^n) \in \typeclass{\cX \times \cY}{n}{\Q_{X} \cdot \P_{Y|X}}} \Prob_{Y^n | X^n, \Theta}(y^n |x^n, \theta) \ge \frac{1}{10^4} 2^{-n\left(\klc{\Q_{Y|X}}{\P^\theta_{Y|X}}{\Q_X} + |\cX \times \cY|\frac{\log{2n}}{n}\right)}.
		\end{align*}
		Substituting for $2 \mu_n$ and using the using the above bound, we get,
		\begin{align*}
			& 2 \cdot 2^{-n \cdot (\expp - \tau')} \ge 2 \mu_n \ge \frac{1}{10^4} 2^{-n \left(\expp - \tau + |\cX \times \cY|\frac{\log{2n}}{n}\right)}.
		\end{align*}
		Since $\tau > \tau'$, for large values of $n$, this is a contradiction.
		Thus we have established, as required by the claim, that for all $x^n \in R$,
		\begin{align*}
			\frac{\sum_{y^n : (x^n, y^n) \in \typeclass{\cX \times \cY}{n}{\Q_{XY}}} \Prob(\binfun_A = \theta | x^n, y^n)}{|\{y^n : (x^n, y^n) \in \typeclass{\cX \times \cY}{n}{\Q_{XY}}\}|} \ge \frac{80}{100}.
		\end{align*}
	\end{proof}
	If there exists $\Q_{XY}$ that satisfies the inequalities in~\eqref{eqn:char-0}, then by Claim~\ref{clm:binfun-seq-corr}, for large enough $n$, there exists $x^n$ such that for $\theta = 0$ and $1$,
	\begin{align*}
		\frac{\sum_{y^n : (x^n, y^n) \in \typeclass{\cX \times \cY}{n}{\Q_{XY}}} \Prob(\binfun_A = \theta | x^n, y^n)}{|\{y^n : (x^n, y^n) \in \typeclass{\cX \times \cY}{n}{\Q_{XY}}\}|} \ge \frac{99}{100}.
	\end{align*}
	This is a contradiction, proving the necessity of condition (i) in the theorem.

	If there exists $\Q_{XY}$ that satisfies the inequalities in~\eqref{eqn:char-2-A}, by Claim~\ref{clm:binfun-seq-priv}, for large enough $n$, there exists $x^n \in \Q_X$ such that for $\theta = 0$ and $1$,
	\begin{align*}
		\frac{\sum_{y^n : (x^n, y^n) \in \typeclass{\cX \times \cY}{n}{\Q_{XY}}} \Prob(\binfun_A = \theta | x^n, y^n)}{|\{y^n : (x^n, y^n) \in \typeclass{\cX \times \cY}{n}{\Q_{XY}}\}|} \ge \frac{80}{100}.
	\end{align*}
	This is also a contradiction, proving the necessity of Condition (ii).

	If there exists $\Q_{XY}$ that satisfies the inequalities in~\eqref{eqn:char-1-A}, there exists $x^n$ and large enough $n$ such that by Claim~\ref{clm:binfun-seq-corr},
	\begin{align*}
		\frac{\sum_{y^n : (x^n, y^n) \in \typeclass{\cX \times \cY}{n}{\Q_{XY}}} \Prob(\binfun_A = \theta | x^n, y^n)}{|\{y^n : (x^n, y^n) \in \typeclass{\cX \times \cY}{n}{\Q_{XY}}\}|} \ge \frac{99}{100},
	\end{align*}
	and by Claim~\ref{clm:binfun-seq-priv},
	\begin{align*}
		\frac{\sum_{y^n : (x^n, y^n) \in \typeclass{\cX \times \cY}{n}{\Q_{XY}}} \Prob(\binfun_A = 1 - \theta | x^n, y^n)}{|\{y^n : (x^n, y^n) \in \typeclass{\cX \times \cY}{n}{\Q_{XY}}\}|} \ge \frac{80}{100}.
	\end{align*}
	This is again a contradiction.
	This proves the necessity of Condition (iv).

	Note that the claims work only for distributions $\Q_{XY}$ with rational \emph{p.d.f}.
	But for $\Q_{XY}$ with irrational \emph{p.d.f}, we may appeal to continuity of KL divergence to get a distribution $\Q'_{XY}$ with rational \emph{p.d.f} that is arbitrarily close to $\Q_{XY}$.
	This proves the lemma.
\end{proof}

\section{Achievability: Missing Proofs}

\begin{proof}[Proof of Claim~\ref{clm:ach}]
For our sequence of decision functions $(\binfun_A^n, \binfun_B^n)_{n \in \bbN}$, we derive $(\delta_n, \mu_n)_{n \in \bbN}$ such that for all $n \in \bbN$ and all $\theta \in \{0,1\}$,
\begin{align}
	&\Prob_{\binfun_A^n | \Theta} (1 - \theta | \theta) \le \delta_n,\; \Prob_{\binfun_B^n | \Theta} (1 - \theta | \theta) \le \delta_n,\label{eqn:correctness}\\
	&\Prob\left(\tvd{\dist{Y^n | x^n}}{\dist{Y^n | x^n, \binfun_A^n, \theta}} \ge \mu_n \right) \le \mu_n, \forall x^n,\label{eqn:privacy_A}\\
	&\Prob\left(\tvd{\dist{X^n | y^n}}{\dist{X^n | y^n, \binfun_B^n, \theta}} \ge \mu_n \right) \le \mu_n, \forall y^n,
\end{align}
and
\begin{align}
	\expc = \lim \limits_{n \rightarrow \infty} -\frac{1}{n} \log{\delta_n},\label{eqn:ach-corr}\\
	\expp = \lim \limits_{n \rightarrow \infty} -\frac{1}{n} \log{\mu_n}.\label{eqn:ach-priv}
\end{align}
	$\binfun_A^n$ is defined such that for $\theta = 0, 1$,
	\begin{align*}
		\Prob_{\binfun_A | \Theta} (1 - \theta | \theta) \le \Prob_{X^n, Y^n | \Theta = \theta}((X^n, Y^n) \in \typeclass{\cX \times \cY}{n}{\Q_{XY}} \text{ s.t. } \kl{\Q_{XY}}{\P^\theta_{XY}} > \expc | \theta). 	
	\end{align*}
	We now appeal to a weak version of Sanov's theorem which we quote verbatim from \cite{CoverT06}.
	\begin{thm}[Theorem 12.2.1, \cite{CoverT06}]\label{thm:weak-sanov}
		Let $X_1, \ldots, X_n$ be \emph{i.i.d.} $\sim \P_X$. Then, for $\epsilon > 0$,
		\begin{align*}
			\Prob_{X^n}(\kl{\type{X^n}}{\P_X} > \epsilon) \le 2^{-n(\epsilon - |\cX|\frac{\log{(n+1)}}{n})}.
		\end{align*}
	\end{thm}
	Hence, we have,
	\begin{align*}
		\Prob_{\binfun_A | \Theta} (1 - \theta | \theta) \le 2^{-n(\expc - |\cX \times \cY|\frac{\log{(n+1)}}{n})}.
	\end{align*}
	Using the same argument, we can show that,
	\begin{align*}
		\Prob_{\binfun_B | \Theta} (1 - \theta | \theta) \le 2^{-n(\expc - |\cX \times \cY|\frac{\log{(n+1)}}{n})}.
	\end{align*}
	Statement~\eqref{eqn:ach-corr} follows directly from this observation.

	To prove~\eqref{eqn:ach-priv}, we analyze $\binfun_A$, the analysis of $\binfun_B$ is done similarly.
	We split the analysis into three cases based on the type of $x^n \in \cX^n$; (1) $x^n$ is such that $\kl{\type{x^n}}{\P^0_X}, \kl{\type{x^n}}{\P^{1}_X} > \expc$; (2) $x^n$ is such that $\kl{\type{x^n}}{\P^\theta_X} \le \expc$ and $\kl{\type{x^n}}{\P^{1 - \theta}_X} > \expc$ for $\theta = 0$ or $1$; (3) $x^n$ is such that $\kl{\type{x^n}}{\P^0_X}, \kl{\type{x^n}}{\P^1_X} \le \expc$.

	If $\kl{\type{x^n}}{\P^\theta} > \expc$ for both $\theta = 0,1$, then $\binfun_A^n(x^n, y^n) = 0$ for all $y^n \in \cY^n$ by the definition of $\binfun_A^n$. Hence, in this case, the decision is independent of the value of $y^n$. From this, it follows that the distributions $\dist{Y^n | x^n, \Theta = \theta}$ and $\dist{Y^n | \binfun(X^n, Y^n) = 0, x^n, \Theta = \theta}$ are identical.

	Case (2) can also be analyzed similarly. If $x^n$ is such that $\kl{\type{x^n}}{\P^0} \le \expc$ and $\kl{\type{x^n}}{\P^1} > \expc$, then $\binfun_A^n(x^n, y^n) = 0$ for all $y^n \in \cY^n$ by the definition of $\binfun_A^n$. Hence, in this case also, the decision is independent of the value of $y^n$. From this, it follows that the distributions $\dist{Y^n | x^n, \Theta = \theta}$ and $\dist{Y^n | x^n, \binfun(X^n, Y^n) = 0, \Theta = \theta}$ are identical.	If $x^n$ is such that $\kl{\type{x^n}}{\P^1} \le \expc$ and $\kl{\type{x^n}}{\P^0} > \expc$, then $\binfun_A^n(x^n, y^n) = 1$ for all $y^n \in \cY^n$ by the definition of $\binfun_A^n$, hence the same analysis applies.

	For case (3), we show that for all $x^n$,
	\begin{align*}
		\Prob(\tvd{\dist{Y^n | x^n, \Theta = 0}}{\dist{Y^n | x^n, \binfun_A^n(X^n, Y^n), \Theta = 0}}) \ge \mu_n) \le \mu_n,
	\end{align*}
	such that $(\mu_n)_{n \in \mathbb{N}}$ satisfies Condition~\eqref{eqn:ach-priv}.
	The case where $\Theta = 1$ can be shown similarly.
	For $i = 0, 1$, define
	\begin{align*}
		S_i = \{y^n \in \cY^n : \Prob_{Y^n | X^n, \Theta}(y^n | x^n, 0) > \Prob_{Y^n | X^n, \binfun_A^n(X^n, Y^n), \Theta}(y^n | x^n, i, 0)\}.
	\end{align*}
	Since $\binfun_A^n$ is a deterministic function, for all $y^n \in \cY^n$
	\begin{align}\label{eqn:acc-rej}
		\Prob_{Y^n | X^n, \binfun_A^n(X^n, Y^n), \Theta}(y^n | x^n, i, 0) = \begin{cases}
			\frac{\Prob_{Y^n | X^n, \Theta}(y^n | x^n, 0)}{\sum_{y^n : \binfun_A^n(x^n, y^n) = i}\Prob_{Y^n | X^n, \Theta}(y^n | x^n, 0)}, \forall y^n \text{ s.t. } \binfun_A^n(x^n, y^n) = i,\\
			0,  \forall y^n \text{ s.t. } \binfun_A^n(x^n, y^n) = 1 - i.
		\end{cases}
	\end{align}
	Hence $y^n \in S_0$ only if $\binfun_A^n(x^n, y^n) = 1$.
	Then, by expanding the total variation distance, we get
	\begin{align}
		& \tvd{\dist{Y^n | x^n, \Theta = 0}}{\dist{Y^n | x^n, \binfun_A^n(X^n, Y^n) = 0, \Theta = 0}}\nonumber\\
		= & \sum_{y^n \in S_0} \left| \Prob_{Y^n | X^n, \Theta}(y^n | x^n, 0) - \Prob_{Y^n | X^n, \binfun_A^n(X^n, Y^n), \Theta}(y^n | x^n, 0, 0)\right|\nonumber\\
		\stackrel{(a)}{\le} & \sum_{y^n : \binfun_A(x^n, y^n) = 1} \left| \Prob_{Y^n | X^n, \Theta}(y^n | x^n, 0) - \Prob_{Y^n | X^n, \binfun_A^n(X^n, Y^n), \Theta}(y^n | x^n, 0, 0)\right|\nonumber\\
		\stackrel{(b)}{\le} & \sum_{y^n : \binfun_A(x^n, y^n) = 1} \Prob_{Y^n | X^n, \Theta}(y^n | x^n, 0).\label{eqn:ach-tvd-bound1}
	\end{align}
	Here, (a) is true since $y^n \in S_0$ only if $\binfun_A^n(x^n, y^n) = 1$ and (b) follows from~\eqref{eqn:acc-rej}.

	By the definition of $\binfun_A^n$, when $x^n$ is such that $\kl{\type{x^n}}{\P^0} \le \expc$, $\binfun_A^n(x^n, y^n) = 1$ only if $\klc{\Q_{Y|X}}{\P^0_{Y|X}}{\Q_X} > \expp$, where $\type{(x^n, y^n)} = \Q_{XY}$.
	Define subset of types $C = \{\Q_{XY} \in \types{\cX \times \cY}{n}: \type{x^n} = \Q_X \text{ and } \klc{\Q_{Y|X}}{\P^0_{Y|X}}{\Q_X} > \expp\}$. Then,
	\begin{align}
		& \sum_{y^n : \binfun_A(x^n, y^n) = 1} \Prob_{Y^n | X^n, \Theta}(y^n | x^n, 0) \le \sum_{\Q_{XY} \in C} \sum_{y^n : \type{(x^n, y^n)} = \Q_{XY}} \Prob_{Y^n | X^n, \Theta}(y^n | x^n, 0) \nonumber\\
		\stackrel{(a)}{\le} & \sum_{\Q_{XY} \in C} 2^{-n \left(\klc{Q_{Y|X}}{P^0_{Y|X}}{\Q_{X}}\right)} \stackrel{(b)}{\le} (n + 1)^{|\cX \times \cY|} \cdot 2^{-n \expp}.\label{eqn:ach-tvd-bound2}
	\end{align}
	Here, (a) follows from Lemma~\ref{lem:kld-prob}, and (b) follows from the fact that for all $\Q_{XY} \in C$, $\klc{\Q_{Y|X}}{\P^0_{Y|X}}{\Q_X} > \expp$ and $|C| \le |\types{\cX \times \cY}{n}| \le (n + 1)^{|\cX \times \cY|}$. Hence we have,
	From the inequalities~\eqref{eqn:ach-tvd-bound1}~and~\eqref{eqn:ach-tvd-bound2}, we get
	\begin{align*}
		& \tvd{\dist{Y^n | x^n, \Theta = 0}}{\dist{Y^n | x^n, \binfun_A^n(X^n, Y^n) = 0, \Theta = 0}} \le
		(n + 1)^{|\cX \times \cY|} \cdot 2^{-n \expp}.
	\end{align*}
	Hence,
	\begin{multline*}
		\Prob\left(\tvd{\dist{Y^n | x^n, \Theta = 0}}{\dist{Y^n | x^n, \binfun_A^n(X^n, Y^n), \Theta = 0}} \ge 2 (n + 1)^{|\cX \times \cY|}\cdot 2^{-n \expp}\right)\\
		\le \Prob_{\binfun_A(X^n, Y^n) | X^n, \Theta}(1 | x^n, 0) \stackrel{(a)}{=} \sum_{y^n : \binfun_A(x^n, y^n) = 1} \Prob_{Y^n | X^n, \Theta}(y^n | x^n, 0)
		\stackrel{(b)}{\le} (n + 1)^{|\cX \times \cY|} \cdot 2^{-n \expp}.
	\end{multline*}
	In (a), we used the fact that $\binfun_A$ is a deterministic function of $x^n, y^n$ and (b) is already shown in~\eqref{eqn:ach-tvd-bound2}.
	This proves the claim when we set $\mu_n = 2(n+1)^{|\cX \times \cY|} \cdot 2^{-n\expp}$.

\end{proof}

\begin{proof}[Proof of Theorem~\ref{mpc-ot}]
	$\prot$ is essentially a \emph{perfectly secure protocol} that computes $\binfun_A(x^n, y^n)$ and $\binfun_B(x^n, y^n)$ at Alice and Bob, respectively, when Alice and Bob have $x^n$ and $y^n$, respectively, as inputs.
	A well known result in secure multi-party computation states that two parties can \emph{securely compute} any pair of (possibly randomized) functions of their combined inputs provided that they have access to sufficiently many copies of OT correlations~\cite[Section 2.4]{EKRTextbook}~and~\cite{Kilian88}.
	We state a version of this fact as the following theorem.
	We will show that the protocol $\prot$ described in the theorem below satisfies the conditions in our claim.
	\begin{thm}
		Let $\binfun_A, \binfun_B : \cX \times \cY \rightarrow \{0, 1\}$ be a pair of randomized boolean functions. For sufficiently large $k$, when $W_A, W_B$ consists of $k$ copies of OT correlations, there exists a protocol $\prot$, with the following guarantees.
		\begin{align*}
			\dist{\dec_A | X = x, Y = y} \equiv \dist{\binfun_A(x, y)}, \text{ and }
			\dist{\dec_B | X = x, Y = y} \equiv \dist{\binfun_B(x, y)}, \forall x, y,\\
			\dist{V_A | X = x, \dec_A = i} \equiv \dist{V_A | X = x, \dec_A = i, Y = y}, \forall x \in \cX, i \in \{0,1\} \text{ and } y \in \cY \text{ s.t. } \Prob_{\binfun_A(X, Y) | X, Y} (i | x, y) > 0,\\
			\dist{V_B | Y = y, \dec_B = i} \equiv \dist{V_B | Y = y, \dec_B = i, X = x}, \forall x \in \cX, i \in \{0,1\} \text{ and } y \in \cY \text{ s.t. } \Prob_{\binfun_A(X, Y) | X, Y} (i | x, y) > 0.
		\end{align*}
	\end{thm}
	Let $v_A \in \cV_A$ such that $\decfun_A(v_A) = i$, and $x$ and $w_A$ are part of the view $v_A$, then
	\begin{align*}
		\Prob_{Y | V_A, X, W_A, \Theta}(y | v_A, x, w_A, \theta) & = \Prob_{Y | V_A, X, W_A, \dec_A, \Theta}(y | v_A, x, w_A, i, \theta)\\
		& = \frac{\Prob_{Y | X, \dec_A, \Theta}(y | x, i, \theta) \cdot \Prob_{V_A, W_A | Y, X, \dec_A, \Theta}(v_A, w_A | y, x, i, \theta)}{\Prob_{V_A, W_A | X, \dec_A, \Theta}(v_A, w_A | x, i, \theta)}\\
		& = \frac{\Prob_{Y | X, \dec_A, \Theta}(y | x, i, \theta) \cdot \Prob_{V_A | Y, X, \dec_A, \Theta}(v_A | y, x, i, \theta)}{\Prob_{V_A | X, \dec_A, \Theta}(v_A | x, i, \theta)}\\
		& \stackrel{(a)}{=} \Prob_{Y | X, \dec_A, \Theta}(y | x, i, \theta) = \Prob_{Y | X, \binfun_A, \Theta}(y | x, i, \theta)
	\end{align*}
	The above theorem guarantees that for all $v_A, x, y$ and $i$, $\Prob_{V_A | Y, X, \dec_A, \Theta}(v_A | y, x, i, \theta) = \Prob_{V_A | X, \dec_A, \Theta}(v_A | x, i, \theta)$, (a) follows from this observation.
	This proves the theorem.
\end{proof}

\end{appendices}
\end{document}